\documentclass[10pt]{article}
\usepackage[a4paper,centering,vscale=0.732,hscale=0.712]{geometry}
\usepackage{
  amsfonts,
  amssymb,
  mathtools,
  amsthm,
  mathrsfs,
}

\newtheorem{prop}{Proposition}[section]
\newtheorem{thm}[prop]{Theorem}
\newtheorem{lemma}[prop]{Lemma}
\newtheorem{coroll}[prop]{Corollary}
\theoremstyle{definition}
\newtheorem{defi}[prop]{Definition}
\theoremstyle{remark}
\newtheorem{rmk}[prop]{Remark}
\newtheorem{rmks}[prop]{Remarks}

\newcommand{\cD}{\mathscr{D}}
\newcommand{\E}{\mathop{{}\mathbb{E}}\nolimits}
\newcommand{\cF}{\mathscr{F}}

\newcommand{\enne}{\mathbb{N}}
\renewcommand{\P}{\mathbb{P}}
\newcommand{\Q}{\mathbb{Q}}
\newcommand{\erre}{\mathbb{R}}

\renewcommand{\geq}{\geqslant}
\renewcommand{\leq}{\leqslant}

\numberwithin{equation}{section}

\DeclarePairedDelimiter\abs{\lvert}{\rvert}
\DeclarePairedDelimiter\norm{\lVert}{\rVert}
\DeclarePairedDelimiterX\ip[2]{\langle}{\rangle}{#1,#2}

\mathtoolsset{centercolon}

\title{On certain representations of pricing functionals}
\author{Carlo Marinelli%
  \thanks{Department of Mathematics, University College London, Gower
    Street, London WC1E 6BT, United Kingdom. URL: \texttt{goo.gl/4GKJP}}}
\date{\normalsize September 12, 2021}

\begin{document}
\maketitle
\begin{abstract}
  We revisit two classical problems: the determination of the law of
  the underlying with respect to a risk-neutral measure on the basis
  of option prices, and the pricing of options with convex payoffs in
  terms of prices of call options with the same maturity (all options
  are European). The formulation of both problems is expressed in a
  language loosely inspired by the theory of inverse problems, and
  several proofs of the corresponding solutions are provided that do
  not rely on any special assumptions on the law of the underlying and
  that may, in some cases, extend results currently available in the
  literature.
\end{abstract}

\newif\ifbozza
\bozzafalse

\section{Introduction}
Let $S$, $\beta \colon \Omega \times [0,T] \to \erre_+$ denote the
price processes of an asset and of a num\'eraire (that we shall assume
to be the money-market account, for simplicity), respectively, in an
arbitrage-free market, modeled on a filtered probability space
$(\Omega,\cF,(\cF_t)_{t\in[0,T]},\P)$, where $T>0$ is a fixed time
horizon and $\cF_0$ is the trivial $\sigma$-algebra. Assuming that
pricing takes place with respect to a risk-neutral probability measure
$\Q$, the price at time zero of a European option with maturity $T$
and payoff profile $g\colon \erre_+ \to \erre$ on the asset with price
process $S$ is given by
\[
\pi(g) = \E_\Q \beta_T^{-1} g(S_T).
\]
We shall call the map $g \mapsto \pi(g)$, defined on the set of all
measurable functions $g$ such that the right-hand side is finite, the
pricing functional.

A rather general and natural question, of clear relevance also for
practical purposes, is the following: suppose that the action of $\pi$
is known on a set of functions $G$, i.e. that $\pi(g)$ is known for
every $g \in G$. Is it possible to enlarge the set of functions $G$
where $\pi$ is determined, i.e. to compute $\pi(f)$ for some functions
$f$ that do not belong to $G$? We are going to discuss some questions
of this type (although not in this generality), through the
representation of the pricing functional as a (Stieltjes) measure,
that is
\[
\pi(g) = \int_{\erre_+} g(x)\,dF(x),
\]
where $F$ is the (right-continuous version of the) distribution
function of $S_T$ with respect to the measure
$(d\Q/d\P)\beta_T^{-1} \cdot \P$, i.e. the measure with density with
respect to $\P$ equal to the stochastic discount factor.

If a collection $G$ of payoff profiles $g$ is such that the prices
$\pi(g)$ of the corresponding options are known, the set
$M:=(g,\pi(g))_{g \in G}$ will be called a measurement set, and a
measurement set that determines $F$ will be called a
representation. That is to say, if knowing $M$ allows one to
reconstruct $F$, then knowing $M$ is equivalent to knowing the pricing
functional itself, which is why we say that it is a representation (of
$dF$, or of $\pi$).  An interesting and important example of a
representation is given by prices of put options: if $G$ is composed
of all functions $g_k\colon x \mapsto (k-x)^+$, $k \in \erre_+$, then
$M$ as defined above is a representation. More precisely, if
$P(k)=\pi(g_k)$ denotes the price at time zero of the put option with
maturity $T$ and strike $k \geq 0$, then $D^+P(k)=F(k)$ for every $k
\in \erre_+$ (see~\S\ref{sec:PCF} and \S\ref{sec:distrib} below). The
problem of reconstructing the law of the underlying from option prices
was probably considered first in \cite{BreeLitz}, where the authors
showed that, denoting the price of a call option with maturity $T$ and
strike $k$ by $C(k)$, the second derivative of $C$ is the density of
$S_T$ with respect to the measure $(d\Q/d\P)\beta_T^{-1} \cdot \P$,
i.e. the first derivative of $F$ (all involved functions are
implicitly assumed to be sufficiently regular in
\cite{BreeLitz}). This result, known as the Breeden-Litzenberger
formula, has found many applications, e.g. in static hedging,
non-parametric density estimation, and local volatility models (see,
e.g., \cite{AL,Boss,Itkin} and references therein, as well as \cite{TaVii},
where an interesting extension to the multidimensional setting is
presented and further references regarding static hedging are given).
In general, $F$ is not of class $C^1$, hence $C$ is not of class
$C^2$, but, if they are, then the formula in \cite{BreeLitz} follows
immediately from $D^+P=F$ and the put-call parity relation. Note,
however, that for pricing purposes it suffices to determine $F$ rather
than its derivative, and the relation $D^+P=F$ is obtained here
without any a priori assumptions on $F$.

The reconstruction of $F$ from a set of option prices is interesting
in its own right, but sometimes less information is enough for the
problem at hand (roughly speaking, this is just the idea behind static
hedging). Using the above terminology, if one needs a measurement set
$M$, it may be possible to determine a measurement set $M'$ that
contains $M$, without necessarily recovering $F$ first. The simplest
example is the pricing of options with continuous piecewise linear
payoff profile. Another one is the pricing of options with payoff
function equal to the difference of convex functions in terms of call
options. Even though, in the latter case, the measurement set of all
call options is already a representation, there is an alternative
pricing formula that avoids the differentiation of $C$, which might be
preferable for numerical purposes. Such pricing formula for options
with convex payoff profiles is not new, but we give nonetheless
several proofs: a very concise one, a longer one that (hopefully)
highlights the role of convexity, and a third one that is extremely
simple if sufficient regularity is present. We also show that similar
ideas can be used to ``localize'' the pricing formula, i.e. to price
options with payoff profile that is piecewise the difference of convex
functions.

The main content is organized as follows: we collect in
\S\ref{sec:prel} some useful (elementary) facts from measure theory,
convexity, and the theory of distribution.
Definitions, motivations, basic properties, and examples pertaining to
pricing functionals, measurement sets, and representations are given
in \S\ref{sec:pf}.
Qualitative properties of the functions $P$ and $C$, as defined above,
are discussed in \S\ref{sec:PCF}, without any assumption on
$F$. Moreover, we show that $F$ is the right derivative of $P$ by two
methods, that is, using the integration by parts formula for
c\`adl\`ag functions of finite variation and by a denseness argument,
respectively.
In \S\ref{sec:conv} we revisit the fact that prices of options with
convex profile are determined by prices of call options for all
positive strikes. This is proved in two ways: by an integration
argument, that uses essentially only the Fubini theorem, and via the
above-mentioned integration by parts formula.
The results of the previous two sections are derived by yet another
approach in \S\S\ref{sec:PCF}-\ref{sec:conv}, that is, using the
theory of distributions. An interesting aspect of this method is that
it provides a particularly handy way to make computations, also in
cases that do not directly follow from the setups of the previous two
sections.
We conclude in \S\ref{sec:X} considering a kind of representation
where a sequence of measures converging towards $F$ intervenes. This
is motivated by an empirical problem on non-parametric pricing of
(vanilla) options treated in \cite{cm:Herm}.

\ifbozza\newpage\else\fi
\section{Preliminaries}
\label{sec:prel}
We shall use some elementary facts from measure theory and convexity,
that we recall for convenience.
Let $(X,\mathscr{A})$ and $(Y,\mathscr{B})$ be measurable spaces, and
$\mu$ a measure on the former. If $\phi \colon X \to Y$ is a
measurable function, then the image measure or push-forward of $\mu$
through $\phi$ is the measure on $(Y,\mathscr{B})$ defined by
$\phi_*\mu \colon B \mapsto \mu(\phi^{-1}(B))$. If $g \colon Y \to
\overline{\erre}$ is a measurable function, then
\begin{equation}
  \label{eq:cava}
  \int_Y g\,d\phi_*\mu = \int_X g \circ \phi \,d\mu,
\end{equation}
in the sense that $g$ is $\phi_*\mu$-integrable if and only if
$g \circ \phi$ is $\mu$-integrable, and in this case the integrals
coincide (see, e.g., \cite[\S2.6.8]{Cohn:MT}).  Interpreting
precomposition as pull-back, hence writing
$\phi^\ast g := g \circ \phi$, and using the notation
$m(f) := \ip{m}{f} := \ip{f}{m} := \int f\,dm$ for any function $f$
integrable with respect to a measure $m$, the identity \eqref{eq:cava}
can be written in the simple and suggestive form
\[
\ip{\phi^\ast g}{\mu} = \ip{g}{\phi_\ast\mu}.
\]

\smallskip

We shall extensively use an integration-by-parts formula for
Lebesgue-Stieltjes integrals. Let the functions $F$, $G\colon \erre
\to \erre$ be c\`adl\`ag (i.e. right-continuous with left limits) and
with finite variation (i.e. having bounded variation on every bounded
interval). Then, for any two real numbers $a<b$,
\begin{equation}
  \label{eq:ibp}
  F(b)G(b) - F(a)G(a) = \int_{\mathopen]a,b\mathclose]} G(x-)\,dF(x) +
  \int_{\mathopen]a,b\mathclose]} F(x)\,dG(x).
\end{equation}
If $G$ is continuous, one can obviously replace $G(x-)$ by
$G(x)$. Whenever $dG$ is an atomless measure we shall just write
$\int_a^b F\,dG$ instead of denoting the interval of integration as a
subscript. Moreover, we set $\erre_+ := [0,+\infty\mathclose[$.

\smallskip

Let $I \subset \erre$ be an open interval and $f \colon I \to \erre$
be a convex function. Then $f$ is everywhere left- and
right-differentiable, that is, for any $x \in I$ the left and right
derivatives
\[
  D^-f(x) := \lim_{h \to 0-} \frac{f(x+h)-f(x)}{h}, \qquad
  D^+f(x) := \lim_{h \to 0+} \frac{f(x+h)-f(x)}{h}
\]
exists and are finite, and $D^-f(x) \leq D^+f(x)$. Moreover, both
$D^-f$ and $D^+f$ are increasing functions, $D^-f$ is left-continuous,
and $D^+f$ is right-continuous. It follows that $f$ is differentiable
except at the countable set of points where $D^-f$ and $D^+f$ do not
coincide. The subdifferential of $f$ at $x$ is defined as
\[
\partial f(x) = \bigl\{ z \in \erre : f(y) - f(x) \geq z(y-x) 
\quad \forall y \in I \bigr\}.
\]
It can be shown that $\partial f(x)=[D^-f(x),D^+f(x)]$ and that, for
any $x_1,x_2 \in I$, $x_1 < x_2$, it holds $D^+f(x_1) \leq D^-f(x_2)$,
hence $\partial f(x_1) \cap \partial f(x_2)$ is either empty, if the
last inequality is strict, or equal to $\{D^-f(x_2)\}$, if the last
inequality is an equality.
The right derivative $D^+f$, being increasing, hence of bounded
variation, and right-continuous, defines a (Lebesgue-Stieltjes)
measure $m$ via the prescription
\[
D^+f(b)-D^+f(a) =: m(\mathopen]a,b]), \qquad a,b \in I, \; a \leq b.
\]
In this sense, the positive measure $m$ can be interpreted as the
second derivative of $f$. For a proof of these results on convex
functions see, e.g., \cite[Chapter~1]{Simon:Conv}.

\smallskip

We shall also use elementary properties of distributions, for which we
refer to, e.g., \cite{Schwartz:ThD}.  Assume that
$f \colon \erre \to \erre$ is piecewise of class $C^1$, with
discontinuity points $(x_n)$. Using the standard notation
$\Delta f(x) := f(x+)-f(x-)$, one has (see
\cite[p.~37]{Schwartz:ThD})
\[
  f' = \sum_n \Delta f(x_n)\delta_{x_n} + [f'],
\]
where $f'$ stands for the derivative of $f$ in the sense of
distributions, and $[f']$ for the derivative in the classical sense
over the open intervals $\mathopen]x_n,x_{n+1}\mathclose[$ where $f$
is continuously differentiable (a corresponding result for
higher-order derivatives can be obtained by induction). If $f$ is a
function with finite variation, then $f'$ coincides, in the sense of
distributions, with the Lebesgue-Stietjes measure $df$. We shall need
to consider functions $f$ that are piecewise differences of convex
functions, i.e. of the form
\[
f = \sum_n f_n, \qquad f_n \colon [a_n,a_{n+1}\mathclose[ \to \erre,
\]
where the sum is (at most) countable, and for every $n$ there exist
convex functions $h^1_n$, $h^2_n$ on $\erre$ such that
$f_n = h^1_n-h^2_n$ on $[a_n,a_{n+1}\mathclose[$. Then $f$ is
c\`adl\`ag and has finite variation. We are going to compute the first
and second distributional derivatives of $f$. It is clear that it is
enough to consider, without loss of generality, $f$ with support equal
to $[a,b]$ and $f=h^1-h^2$ on $[a,b\mathclose[$, for $h^1$ and $h^2$
convex functions on $\erre$.
\begin{lemma}
  Let $f\colon \erre \to \erre$ be a c\`adl\`ag function with finite
  variation, and $[a,b] \subset \erre$ a compact interval. The
  distributional derivative of $f1_{[a,b]}$ is
  \[
    f' = df\big\vert_{\mathopen]a,b\mathclose[}
    + f(a) \delta_a - f(b-)\delta_b.
  \]
\end{lemma}
\begin{proof}
  Let us denote $f1_{[a,b]}$, for the purposes of this proof only,
  simply by $f$. The distributional derivative $f'$ is defined by the
  identity $\ip{f'}{\phi} = -\ip{f}{\phi'}$ for every $\phi \in
  \cD(\erre)$, where
  \[
    \ip{f}{\phi'} = \int_\erre f \phi' = \int_a^b f\,d\phi
    = \int_{\mathopen]a,b]} f\,d\phi
  \]
  and, thanks to the integration-by-parts formula,
  \[
    f(b)\phi(a) - f(a)\phi(a) = \int_{\mathopen]a,b]} f\,d\phi
    + \int_{\mathopen]a,b]} \phi\,df,
  \]
  hence
  \[
    \int_{\mathopen]a,b]} f\,d\phi
    = - \int_{\mathopen]a,b]} \phi\,df + f(b)\phi(b) - f(a)\phi(a),
  \]
  i.e.
  \begin{align*}
    \ip{f'}{\phi} 
    &= \int_{\mathopen]a,b]} \phi\,df + f(a)\phi(a)
      - f(b)\phi(b)\\
    &= \int_{\mathopen]a,b\mathopen[} \phi\,df + \phi(b)(f(b)-f(b-))
      + f(a)\phi(a) - f(b)\phi(b)\\
    &= \int_{\mathopen]a,b\mathopen[} \phi\,df + f(a)\phi(a) - f(b-)\phi(b)
    \qedhere
  \end{align*}
\end{proof}

\begin{prop}
  \label{prop:dcd}
  Let $f\colon \erre \to \erre$ be a difference of convex functions
  and let $[a,b] \subset \erre$ be a compact interval. Then the first
  and second distributional derivatives of $f1_{[a,b]}$, denote by
  $f'$ and $f''$, respectively, are
  \begin{align*}
    f' &= df\big\vert_{\mathopen]a,b\mathclose[}
         + f(a) \delta_a - f(b)\delta_b,\\
    f'' &= dD^+f\big\vert_{\mathopen]a,b\mathclose[}
          + f(a)\delta'(a) - f(b)\delta'_b
          + D^+f(a)\delta_a - D^+f(b-)\delta_b.
  \end{align*}
\end{prop}
\begin{proof}
  The function $f$ is continuous on $\erre$ and, being the difference
  of convex functions, has finite variation. The previous lemma then
  yields the expression for $f'$. To compute $f''$, let us recall that
  $f$ is absolutely continuous on with (classical) derivative equal to
  $D^+f$ a.e., so that
  \begin{align*}
    \ip{f''}{\phi} = - \ip{f'}{\phi'}
    &= - \int_{\mathopen]a,b\mathclose[} \phi'\,df - f(a)\phi'(a)
      + f(b-)\phi'(b)\\
    &= - \int_{\mathopen]a,b\mathclose[} D^+f\,d\phi - f(a)\phi'(a)
      + f(b-)\phi'(b).
  \end{align*}
  Since $D^+f$ is c\`adl\`ag and of finite variation, the
  integration-by-parts formula yields
  \[
    D^+f(b)\phi(b) - D^+f(a)\phi(a) = \int_{\mathopen]a,b]} D^+f\,d\phi
    + \int_{\mathopen]a,b]} \phi\,dD^+f,
  \]
  hence
  \begin{align*}
    - \int_{\mathopen]a,b\mathclose[} D^+f\,d\phi
    &= - \int_{\mathopen]a,b]} D^+f\,d\phi\\
    &= \int_{\mathopen]a,b]} \phi\,dD^+f + D^+f(a)\phi(a) - D^+f(b)\phi(b)\\
    &= \int_{\mathopen]a,b\mathclose[} \phi\,dD^+f
    + D^+f(a)\phi(a) - D^+f(b-)\phi(b).
  \end{align*}
  Collecting terms concludes the proof:
  \[
    \ip{f''}{\phi} = \int_{\mathopen]a,b\mathclose[} \phi \,dD^+f
    - f(a) \phi'(a) + f(b-) \phi'(b)
    + D^+f(a) \phi(a) - D^+f(b-) \phi(b).
    \qedhere
  \]
\end{proof}
\begin{rmk}
  Let $a<b$ be real numbers and $I$ be any interval with endpoints $a$
  and $b$. Note that $f1_{[a,b]}$ coincides in $\cD'$ with $f1_I$, for
  any choice of $I$. Therefore their distributional derivatives are
  also the same.
\end{rmk}

\ifbozza\newpage\else\fi
\section{Pricing functionals, measurements and representations}
\label{sec:pf}
\subsection{Pricing functionals}
Let $(\Omega,\mathscr{F},\P)$ be a probability space endowed with a
filtration $(\mathscr{F}_t)_{t\in[0,T]}$, with $T>0$ a fixed time
horizon, and let $S \colon \Omega \times [0,T] \to \erre_+$ be the
price process of an asset. We assume also that $\beta \colon \Omega
\times [0,T] \to \mathopen]0,\infty\mathclose[$ is the price process
of a further asset used as num\'eraire, normalized with $\beta_0=1$
and uniformly bounded from below, and that the market where both
assets are traded is free of arbitrage, so that the set $\mathsf{Q}$
of probability measures $\Q$ equivalent to $\P$ such that the
discounted price process $\beta^{-1}S$ is a $\Q$-local martingale is
not empty. For any $\mathscr{F}_T$-measurable claim $X$ such that
$\beta_T^{-1}X$ is bounded, the value $\E_\Q \beta_T^{-1}X$ is an
arbitrage-free price at time zero of $X$ for every $\Q \in
\mathsf{Q}$. From now on we shall fix a measure $\Q \in
\mathsf{Q}$. For any measurable bounded function $g \colon \erre_+ \to
\erre$, the (bounded) $\mathscr{F}_T$-measurable random variable
$g(S_T)$ is the payoff of a European option on $S$ with payoff profile
$g$, the price of which at time zero is $\E_\Q \beta_T^{-1} g(S_T)$.

We shall call \emph{pricing functional} the map
\[
  \pi \colon g \longmapsto \E_\Q \beta_T^{-1} g(S_T)
  = \E\frac{d\Q}{d\P}\beta_T^{-1} g(S_T),
\]
defined first on the set of measurable bounded functions $g \colon
\erre_+ \to \erre$. 
Let $\mu$ be the measure on $\cF_T$ defined by
\[
\mu(A) := \E_\Q \beta_T^{-1} 1_A,
\]
that is, $\mu$ is the measure on $\cF_T$ the Radon-Nikodym
derivative of which with respect to $\P$ is
\[
\frac{d\mu}{d\P} = \frac{d\Q}{d\P} \beta_T^{-1}.
\]
Note that $\mu$ is (in general) not a probability measure: in fact,
$\mu(\Omega) = \E_\Q \beta_T^{-1}$ need not be one, and could be
interpreted as the price at time zero of a zero-coupon bond maturing
at time $T$ with face value equal to one. In this case, $\mu$ is a
sub-probability measure, i.e. $\mu(\Omega) \leq 1$.
The pricing functional can then be written as
\[
\pi \colon g \longmapsto \int_\Omega g(S_T)\,d\mu.
\]
Denoting the push-forward of $\mu$ through $S_T$ by $S_\ast\mu$, i.e.
the measure on the Borel $\sigma$-algebra of $\erre$ defined by
\[
S_\ast\mu \colon B \longmapsto \mu(S_T^{-1}(B)),
\]
one has
\[
\E_\Q \beta_T^{-1} g(S_T) = 
\int_\Omega g(S_T)\,d\mu = \int_\erre g\,d(S_\ast\mu).
\]
Therefore, denoting the distribution function of the measure
$S_\ast\mu$ by $F$, i.e.
\[
F(x) := \mu(S_T \leq x) = \E_\Q \beta_T^{-1} 1_{\{S_T \leq x\}},
\]
the pricing functional can be written as
\[
\pi \colon g \longmapsto \int_{\erre_+} g\,dF.
\]
In other words, the pricing functional can be identified with $F$, or
with $S_\ast\mu$. Note also that the pricing functional can naturally
be extended to every $g \in L^1(dF)$. If $g \geq 0$, as is mostly the
case for payoff functions, then $\pi(g)$ is simply the norm
of $g$ in $L^1(dF)$.

\begin{rmk}
  If one just assumes that there exists a pricing functional, defined
  as a positive linear functional on a certain set of functions, then
  an integral representation of $\pi$ holds in many cases. This is
  essentially the content of various forms of the Riesz representation
  theorem. For instance, if $\pi$ is continuous on $C_0(\erre)$, the
  Banach space of continuous functions that are zero at infinity,
  endowed with the supremum norm, then there exists a unique finite
  Radon measure $m$ on $\erre$ such that $\pi(g)=\int g\,dm$ for every
  $g \in C_0(\erre)$. If $\pi$ is continuous on $C_c(\erre)$, the
  space of continuous functions with compact support with the topology
  of uniform convergence on compact sets, then there exists a unique
  Radon measure $m$ on $\erre$ (not necessarily finite) such that
  $\pi(g)=\int g\,dm$ for every $g \in C_0(\erre)$. 
  On the other hand, if $\pi$ is just assumed to be continuous on
  $\mathscr{L}^\infty(\erre)$, the Banach space of bounded functions
  with the supremum norm, then an integral representation of $\pi$ is
  only possible with respect to a bounded additive set function, not a
  measure. A completely analogous situation arises if continuity of
  $\pi$ is assumed on $L^\infty(\erre)$. On the other hand, if $\pi$
  is assumed to be weak* continuous on $L^\infty(\erre)$, then there
  exists $\phi \in L^1_+(\erre)$ such that $\pi(g)=\int \phi g$ for
  every $g \in L^\infty(\erre)$. However, the weak* convergence of a
  sequence $(g_n)$ in $L^\infty$, i.e. the existence of $g \in
  L^\infty$ such that
  \[
    \lim_{n \to \infty} \int fg_n = \int fg \quad \forall f \in
    L^1(\erre),
  \]
  does not seem to have a clear economic interpretation.
\end{rmk}

\subsection{Measurements and representations}
Depending on the problem at hand, the pricing functional
$\pi \colon g \mapsto dF(g)$ may or may not be known. If $dF$ is
assumed a priori to be known, for instance in the Black-Scholes model
with given volatility, then $\pi$ is trivially known. Analogously, one
may assume that $dF$ belongs to a certain family of finite measures
$(dF_\theta)_{\theta\in\Theta}$ parametrized by a finite-dimensional
parameter $\theta$, and by statistical procedures an estimate
$\hat{\theta}$ is obtained, so that $dF_{\hat{\theta}}$ is then used
in the definition of $\pi$, thus falling back to the previous (quite
tautological) case. Strictly speaking, this procedure produces an
estimator of the pricing functional, but we are not going to discuss
any issues pertaining to statistics. On the other hand, in many other
situations, for instance when no parametric assumptions on $dF$ are
made, the pricing functional $\pi$ is only known through its action on
a set of ``test functions'' $(g_j)_{j \in J}$, e.g. with $g_j$ the
payoff profile of a call or put option with a strike price indexed by
$j \in J$. The next definition is hence quite natural.
\begin{defi}
  A \emph{measurement} (of $dF$) is a pair $(g,\pi(g))$, where $g
  \colon \erre_+ \to \erre$ is a measurable function integrable with
  respect to $dF$. A \emph{measurement set} (of $dF$) is a collection
  of measurements.
\end{defi}
A typical situation of practical relevance is given by $(g_j)$ being a
collection of payoff profiles of (European) options. For instance, for
any $j \geq 0$, let $g_j$ be the payoff function of a put option with
strike price $j$, that is, $g_j \colon x \mapsto (j-x)^+$.  If the
price of the put option with strike $j$ is known for every $j>0$, then
we have a measurement $M=(g_j,\pi_j)_{j \in J}$ setting $J=\erre_+$,
$g_j \colon x \mapsto (j-x)^+$, and $\pi_j=dF(g_j)$.

\begin{rmks}
  a) Let $M = (g_j,\pi_j)_{j \in J}$ be a measurement set ($J$ is just
  an index set). The set of numbers $(\pi_j)$ is included in the
  definition of $M$ just for convenience, as it is clearly redundant
  being uniquely determined by $(g_j)$ and $dF$, the latter of which
  is assumed to be fixed, even though treated as unknown.
  b) A measurement set is just a subset of the graph of
  $\pi$.
  c) The term ``measurement'', by no means standard, somewhat
  mimics an analogous one used in the theory of inverse problems,
  where, in a (usually) more rigid functional setting, the expression
  ``measurement operator'' is sometimes used.
  d) In view of the linearity of integration, if $\pi$ is known on a
  set $G \subseteq L^1(dF)$, then it is known also on the vector space
  generated by $G$. Similarly, it would seem natural to augment $M$
  with its accumulation points, i.e. to take its closure, in
  $L^1(dF) \times \erre$. However, since we treat $dF$ as unknown,
  this operation would not be plausible. Some accumulation points can
  be added nonetheless, as we shall see below, so long as they are
  constructed without using $dF$ or, more precisely, assuming that all
  is known about $dF$ is (the vector space generated by) $M$.
\end{rmks}

Measurement sets can be ordered by inclusion, hence they can be
compared. If $M$ is a measurement set, the vector space generated by
$M$, itself a measurement set, will be denoted by $\hat{M}$.
\begin{defi}
  Let $M_1$ and $M_2$ be two measurement sets. One says that $M_1$ is
  \emph{finer} than $M_2$ if $\hat{M}_1$ contains $M_2$, and that
  $M_1$ and $M_2$ are \emph{equivalent} if $M_1$ is finer than $M_2$
  and $M_2$ is finer than $M_1$, i.e. if $\hat{M}_1=\hat{M}_2$. A
  \emph{representation} is a measurement set finer than
  ${(1_A,dF(A))}_{A \in \mathscr{A}}$, where $\mathscr{A}$ is any set
  of subsets of $\erre_+$ generating the Borel $\sigma$-algebra
  $\mathscr{B}(\erre_+)$.
\end{defi}
Apart from the natural inverse problem of recovering the measure $dF$
(or equivalently the function $F$) from a sufficiently rich collection
of option prices, possibly providing an algorithm to do so, it is
interesting also to describe relations between measurement sets. For
instance, if one needs $F$ only to price a certain set of options,
instead of reconstructing $F$ it could suffice to identify a
measurement set that already allows to accomplish the task. In the
simplest case, if $g$ is the payoff profile of the option to price and
$g$ belongs to the vector space generated by an available measurement
set $M$, there is clearly no need to recover $F$. In spite of its
simplicity, this is precisely how one can proceed to price options
with continuous piecewise linear payoff profile. In fact, as is well
known, these options can be priced in terms of linear combinations
(independent of $F$!) of prices of put options with strikes at the
``juncture'' points of the piecewise linear profile. A more
sophisticated fact is that call option prices for every positive
strike price allow to price option with arbitrary convex payoff. In
this case, however, if $g$ is an arbitrary convex function and
$\operatorname{pr}_1 M$, the projection on $L^1(dF)$ of the
measurement set $M$, is the vector space generated by $(x \mapsto
(x-k)^+)_{k \in \erre_+}$, it is not true in general that $g \in
\operatorname{pr}_1 M$. It is true, however, that $g$ is an
accumulation point of $\operatorname{pr}_1 M$, as discussed in
\S\ref{sec:conv} below.

It was mentioned above that it would not be meaningful to extend a
measurement set $M$ taking its closure in $L^1(dF) \times \erre$, as
$F$ is considered unknown. However, one can indeed add some cluster
points, if they are defined by procedures that do not involve $F$. In
particular, at least two possibilities exist:
\smallskip\par\noindent
(a) let
$(g_n) \subseteq \operatorname{pr_1} M$ be a sequence that converges
pointwise to $g$ and for which there exists $h \in L^1(dF)$ such that
$\abs{g_n(x)} \leq h(x)$ for all $x \in \erre_+$. The dominated
convergence theorem then implies that $g \in L^1(dF)$ and that
$\pi(g)=\lim_{n \to \infty} \pi(g_n)$;
\smallskip\par\noindent
(b) let $(g_n) \subseteq \operatorname{pr_1} M$ be such that
$g_n \uparrow g$, i.e. $(g_n)$ is an increasing sequence that
convergence pointwise to $g$, and such that $(\pi(g_n))$ is bounded
from above, i.e.  $\sup_n \pi(g_n)<\infty$. Then
$0 \leq g_n-g_0 \uparrow g-g_0$ and, by the monotone convergence
theorem,
\[
  \pi(g-g_0) = \lim_n \pi(g_n-g_0) = \sup_n \pi(g_n) - \pi(g_0) <
  \infty,
\]
hence $g-g_0 \in L^1(dF)$, i.e. $g \in L^1(dF)$ with
$\pi(g) = \sup_n \pi(g_n)$.
\smallskip\par\noindent
The cluster points constructed in (a) and (b) do not depend on knowing
$F$, hence they could reasonably be added to the measurement set
$M$. The measurement sets obtained by adding to $\hat{M}$ the cluster
points described in (a) and (b) will be denoted by $M^d$ and $M^m$,
respectively. We shall see that if $M_1$ is measurement set of all
call options, and $M_2$ the measurement set of all convex options,
then $M_1^m$ is finer than $M_2$. Since $M_2$ is finer than $M^m_1$
(the pointwise supremum of a family of convex functions is convex),
$M_1^m$ and $M_2$ are equivalent measurement sets.  In other words,
one cannot replicate a convex payoff with just call payoffs, but one
can approximate a convex payoff by a combination of call payoffs with
any pricing accuracy.

Taking suitable limits of sequences of measurements is not the only
possible way to enrich a measurement set. In fact, one can also
perform several operations on $(\pi_j)_{j\in J} \subseteq \erre$,
using the structure of $\erre$: they can for instance be added,
multiplied, and functions $\phi\colon \erre^n \to X$ can be applied to
$n$ of them, with $X$ suitable sets, and so on. Note that $M$ could
also be seen as a linear map from the space of finite measures
$\mathscr{M}^1(\erre_+)$ to $\erre^J$, mapping $dF$ to $(dF(g_j))_{j
  \in J}$. Viewing elements of $\erre^J$ as functions from $J$ to
$\erre$, the problem at hand may imply that these functions in the
codomain have additional properties, for instance they may be
monotone, or convex, or continuous, or differentiable, depending on
the inputs $(g_j)$. Depending on the range of $M$ in the codomain
$\erre^J$, different operations may be applied.  For instance, taking
derivatives on $\erre^J$ or on $C(J)$ would not make sense, but it
would make sense on $C^1(J)$, or in the a.e. sense if we knew that the
range is made of Lipschitz continuous functions. We shall see that
this point of view is also fruitful, showing that the right derivative
of put prices, seen as a function $P$ of the strike price, is equal to
$F$. We shall also see that the price of an option with arbitrary
convex payoff can be written in terms of an integral of $C$, where
$C(k)$ is the price of the call option with strike $k$.

In some cases one does not observe a measurement directly, but a
function of a measurement. This is the case, for instance, of implied
volatility. If $g_k \colon x \mapsto (k-x)^+$ is the payoff function
of a put option with strike $k$, there is a one-to-one correspondence
between $\pi_k:=\pi(g_k)$ and the (Black-Scholes) implied volatility,
given by a function $v \colon \erre_+ \to \erre_+$ such that $\pi_k =
\mathsf{BS}(S_0,k,T,v(\pi_k))$.  Here $\mathsf{BS}(S_0,k,T,\sigma)$
denotes the Black-Scholes price (at time zero) of a put option on an
underlying with price at time zero equal to $S_0$, strike $k$, time to
maturity $T$, volatility $\sigma$, and interest rate as well as
dividend rate equal to zero (or to any other constants).  In
particular, if the implied volatility is known for every strike $k>0$,
inverting the function $v$ we obtain the measurement set of put prices
$M=(g_k)_{k \geq0}$, which is a representation.  In other words,
implied volatility for all strikes uniquely determines the pricing
functional or, equivalently, the measure $dF$. We may then say, with a
slight abuse of terminology, that implied volatility is a
representation.

Let $X$ be a locally compact space and $\phi \colon \erre \to X$ be a
measurable isomorphism, i.e. a bijection such that both $\phi$ and
$\phi^{-1}$ are measurable. This is the case, for instance, if $\phi$
is a homeomorphism. Then, for any $g \in L^1(dF)$, one has
\[
  dF(g) = \ip{g}{dF} = \ip[\big]{\phi^\ast(\phi^{-1})^\ast g}{F}
  = \ip[\big]{(\phi^{-1})^\ast g}{\phi_*F}.
\]
This change of parametrization can also be interpreted in terms of
measurement sets, saying that the measurement set $M=(g_j,dF(g_j))$ of
$dF$ is in bijective correspondence with the measurement set
\[
  M' = \bigl( (\phi^{-1})^\ast g_j,\phi_\ast dF((\phi^{-1})^\ast g) \bigr)
     = \bigl( (\phi^{-1})^\ast g_j,dF(g_j) \bigr)
\]
of $\phi_\ast dF$. Even though the two measurements are isomorphic (as
sets), they may have quite different character. Let us consider, for
instance, the reparametrization from price to logarithmic return:
setting $S_T=S_0\exp(\sigma X_T + m)$, where $\sigma>0$ and $m$ are
constants, the pricing functional can be written as
\[
\pi \colon g \longmapsto \int_\erre g(S_0e^{\sigma x + m})dF_X(x),
\]
where $F_X$ is the (right-continuous) distribution function of the
measure $(X_T)_\ast \mu$, the support of which is $\erre$. If $g$ is
the payoff function of a put option with strike $k$, then
$x \mapsto g(S_0e^{\sigma x + m}) = {(k - S_0e^{\sigma x + m})}^+$
does \emph{not} have compact support. This is clearly in stark
contrast to the expression of $\pi(g)$ in terms of $dF$, where the
intersection of the supports of $g$ and $dF$ is compact. As will be
seen, several analytic arguments strongly depend on this property,
that hence cannot be used with the new parametrization, even though
the values of the corresponding integrals are the same.

Finally, we remark that it is sometimes useful to extend the
definition of measurement set adding $(dF_j)$, a collection of
(possibly signed) measures for which a relation to $dF$ is known. For
instance, let $g_k$, for any $k \geq 0$, be the payoff function of a
put option with strike price $k$, that is, $g_k \colon x \mapsto
(k-x)^+$. Moreover, let $(dF_n)_{n\in\enne}$ be a sequence of Radon
measures converging weakly to $dF$ as $n \to \infty$. Each measure
$dF_n$ can be thought of as an approximation to the law $dF$, and
$dF_n(g_k)$ as the price of a put option with strike $k$ under the
approximating law $dF_n$. If all such prices can be observed, then we
have an ``extended'' measurement set $M=(g_j,\pi_j,dF_j)_{j \in J}$,
where $J=\erre_+ \times \enne$, $F_j=F_{kn}$, $F_{kn}=F_n$ for every
$k$, $g_j=g_{kn}=g_k$ for every $n$, and
$\pi_j=\pi_{kn}=F_n(g_k)$. Note that $g_k \in C_b(\erre)$ for every
$k$, hence $F_n(g_k) \to F(g_k)$ as $n \to \infty$. In particular, if
we define the (standard) measurement set $M'=(g_j,\pi_j)_{j \in
  \erre_+}$ as $g_j\colon x \mapsto (j-x)^+$ and $\pi_j=dF(g_j)$, then
we could say that $M$ ``implies'' $M'$. That is, for every $g \in
\operatorname{pr}_1 M'$ there exists a sequence $(dF_n) \subset
\operatorname{pr}_3 M$ such that $\pi(g)=dF(g)$ is the limit of
$dF_n(g) \subset \operatorname{pr}_2 M$. An analogous example
motivated by (empirical) non-parametric option pricing is discussed in
\S\ref{sec:X} below.

\ifbozza\newpage\else\fi
\section{Put and call option prices and the pricing functional}
\label{sec:PCF}
Let us define the numerical functions $P,\,C \colon \erre_+ \to
[0,+\infty]$ by
\[
  P(k) := \int_{\erre_+} (k-x)^+ \,dF(x), \qquad
  C(k) := \int_{\erre_+} (x-k)^+ \,dF(x).
\]
Note that $P(k)$ is finite for all $k$ as $\erre_+ \ni x \mapsto
(k-x)^+$ is bounded (has even compact support), but $C(k)$ is finite
if and only
\[
\int_k^\infty x\,dF(x) < \infty,
\]
hence $C$ is everywhere finite if and only $dF$ has a finite mean
\[
\overline{dF} := \int_{\erre_+} x\,dF(x).
\]
The assumption $\overline{dF}<\infty$ also implies that $dF$ is a
finite measure, hence $F(\infty):=\lim_{x\to\infty} F(x)$ is
finite. In fact, rather obviously, $F(\infty)=F(\infty)-F(1)+F(1)$ and
\[
  F(\infty)-F(1) = \int_{\mathopen]1,\infty\mathclose[} dF
  \leq \int_{\mathopen]1,\infty\mathclose[} x\,dF(x)
  \leq \overline{dF} < \infty.
\]
The financial interpretation of $\overline{dF}<\infty$ is that
$\E_\Q\beta_T^{-1}S_T$ must be finite. This is clearly not a
limitation.  However, we shall mention the hypothesis anyway because
some of the considerations to follow may be interesting for general
$F$, irrespective of the underlying financial interpretation.

The functions $P$ and $C$ will play a central role, so we discuss some
of their properties. They are all rather basic, but they are given
here in full detail because it is simpler to prove them than to look
for a suitable reference.
\begin{prop}
  \label{prop:P}
  The function $P$ is increasing, locally Lipschitz continuous,
  Lipschitz continuous if $dF$ is finite, convex, and satisfies the
  inequality $P(k) \leq kF(k)$ for every $k \geq 0$. Moreover,
  $P(0)=0$ and
  \[
    \lim_{k \to \infty} \frac{P(k)}{k} = F(\infty),
  \]
  hence, in particular, $\lim_{k \to \infty} P(k) = \infty$.
\end{prop}
\begin{proof}
  Since $k \mapsto (k-x)^+$ is increasing for every $x \in \erre_+$
  and integration (with respect to a positive measure) is positivity
  preserving, $P$ is increasing. Similarly, as $k \mapsto k^+$ is
  $1$-Lipschitz continuous, $k \mapsto (k-x)^+$ is $1$-Lipschitz
  continuous uniformly with respect to $x$, hence $P$ is locally
  Lipschitz continuous as well and globally Lipschitz continuous if
  $F(\infty)<\infty$. To prove convexity, note that, for any $x \in
  \erre_+$, $k \mapsto k-x$ is affine, in particular convex, and $y
  \mapsto y^+$ is convex increasing, hence the composite function $k
  \mapsto (k-x)^+$ is convex. Finally, integration with respect to a
  positive measure preserves convexity, hence $P$ is convex. The
  identity $P(0)=0$ follows immediately by the definition of $P$, as
  does the estimate $P(k) \leq kF(k)$, where $F(k) \leq
  F(\infty)$. Finally,
  \[
  \frac{P(k)}{k} = \frac{1}{k} \int_{[0,k]} (k-x)\,dF(x) 
  = \int_{\erre_+} \Bigl( 1 - \frac{x}{k} \Bigr) 1_{[0,k]} \,dF(x),
  \]
  where $(1-x/k)1_{[0,k]} \to 1$ for all $x \geq 0$ as $k \to \infty$
  and $(1-x/k)1_{[0,k]} \in [0,1]$ for all $x,k \geq 0$, hence the
  dominated convergence theorem implies
  \[
    \lim_{k \to \infty} \frac{P(k)}{k} = \int_{\erre_+} dF = F(\infty).
    \qedhere
  \]
\end{proof}

\begin{prop}
  \label{prop:C}
  Assume that $\overline{dF}<\infty$. The function $C$ is decreasing,
  Lipschitz continuous, and convex. Moreover, $C(0)=\overline{dF}$ and
  $\lim_{k \to \infty}C(k)=0$.
\end{prop}
\begin{proof}
  The proof of monotonicity, Lipschitz continuity, and convexity are
  entirely similar to the corresponding proof for put options, noting
  that $k \mapsto (x-k)^+$ is decreasing. The definition of $C$
  immediately implies that $C(0)=\int_{\erre_+}  x\,dF(x)$, and also that
  \[
  C(k) = \int_k^\infty (x-k)\,dF(x) \leq \int_{[k,\infty\mathclose[} x\,dF(x),
  \]
  where the right-hand side converges to zero as $k \to \infty$
  because $\int_{\erre_+} x\,dF(x)$ is finite by assumption.
\end{proof}

By a direct computation one can obtain estimates for local and global
Lipschitz constants. In fact, the $1$-Lipschitz continuity of
$x \mapsto x^+$, hence also of $k \mapsto
(k-x)^+$, yields, for any $k_1,k_2 \geq 0$,
\begin{align*}
  \abs{P(k_2)-P(k_1)} 
  &\leq \int_{[0,k_1 \vee k_2]} \abs[\big]{(k_2-x)^+ - (k_1-x)^+}\,dF(x)\\
  &= \int_{[0,k_1 \vee k_2\mathclose[} \abs[\big]{(k_2-x)^+ - (k_1-x)^+}\,dF(x)\\
  &\leq \int_{[0,k_1 \vee k_2\mathclose[} \abs{k_2 - k_1}\,dF(x)
  = \abs{k_2 - k_1} F_-(k_1 \vee k_2),
\end{align*}
where $F_-$ stands for the left-continuous version of $F$, defined by
$F_-(x) := F(x-) := \lim_{h \downarrow 0} F(x-h)$. The same estimate holds for
$P$ replaced by $C$.
One can actually show, using subdifferentials, that the Lipschitz
continuity estimates thus obtained are sharp. In fact, for any $k_1,
k_2 \geq 0$, convexity implies
\[
  P(k_2) \geq P(k_1) + \partial P(k_1)(k_2-k_1)  
\]
where\footnote{Since $\partial P(k_1)$ is in general a set, one should
  write $P(k_2) \geq P(k_1) + y(k_2-k_1)$ for every $y \in \partial
  P(k_1)$. This slight abuse of notation shall not create any harm
  though.}  $\partial$ stands for the subdifferential in the sense of
convex analysis. Hence, if $k_1 \geq k_2$, $P(k_1) \geq P(k_2)$
because $P$ is increasing, which also implies that $\partial P(x)
\subset \erre_+$ for every $x>0$, hence
\[
\abs[\big]{P(k_1) - P(k_2)} = P(k_1) - P(k_2) 
\leq \partial P(k_1) (k_1-k_2) = \partial P(k_1) \abs{k_1-k_2}.
\]
Similarly, if $k_1 \leq k_2$,
\[
\abs[\big]{P(k_1) - P(k_2)} = P(k_2) - P(k_1)
\leq \partial P(k_2) (k_2-k_1) = \partial P(k_2) \abs{k_1-k_2}.
\]
Recalling that $\partial P(k) = [D^-P(k),D^+P(k)]$ for every $k>0$, it
easily follows that
\[
\abs[\big]{P(k_1) - P(k_2)} \leq D^-P(k_1 \vee k_2) \, \abs{k_1-k_2}.
\]
As $D^+P=F$ and the left-continuous version of $D^+P$ is $D^-P$, it
follows that $D^-P=F_-$.

\medskip

We are going to show that $F$ is the right derivative of $P$, and that
a similar relation holds between the call price function $C$ and $F$.
We give two proofs, one that relies on the integration-by-parts
formula for c\`adl\`ag functions, and one a bit indirect based on a
denseness result: we show that the set of put payoffs are total in
$L^1(dF)$, i.e. that for any $g \in L^1(dF)$ there exist a sequence of
\emph{finite} linear combinations of pay payoffs that converges to $g$
in $L^1(dF)$. This connects with another formulation of representation
that we have discussed, i.e. by a kind of closure operation. Then we
show that the two approaches are in fact equivalent. A third approach,
using distributions, will be given in \S\ref{sec:distrib} below.

\subsection{Reconstruction of $F$ via integration by parts}
We shall apply the integration-by-parts formula to establish formulas
relating the distribution function $F$ and the price functions for put
and call options $P$ and $C$.
\begin{thm}
  One has $P'=F$ a.e. in $\erre_+$ and $D^+P(x)=F(x)$ for every
  $x \in \erre_+$. Moreover, if the measure $dF$ has finite mean, then
  $C'=F-F(\infty)$ a.e. in $\erre_+$ and $D^+C(x)=F(x)-F(\infty)$ for
  every $x \in \erre_+$.
\end{thm}
\begin{proof}
  Let $k \geq 0$ and $G\colon x \mapsto k-x$. The integration-by-parts
  formula
  \[
    G(k)F(k) - G(0)F(0) = \int_{\mathopen]0,k\mathclose]} G(x)\,dF(x)
    + \int_{\mathopen]0,k\mathclose]} F(x)\,dG(x),
  \]
  yields
  \begin{align*}
    \int_0^k F(x)\,dx
    &= kF(0) + \int_{\mathopen]0,k\mathclose]} (k-x)\,dF(x)\\
    &= \int_{[0,k]} (k-x)\,dF(x)\\
    &= \int_{\erre_+} (k-x)^+\,dF(x) = P(k).
  \end{align*}
  The Lebesgue differentiation theorem then implies that $P'=F$
  a.e. in $\erre_+$. Moreover, since $F$ is right-continuous by
  definition, and $P$ is convex, hence right-differentiable, we also
  have $D^+P(x)=F(x)$ for every $x \in \erre_+$.

  Obtaining a relation between $C$ and $F$ along the same lines is a
  bit more involved: if $k>0$ and $G \colon x \mapsto x-k$, one has,
  for any $a>k$,
  \[
    G(a)F(a) - G(k)F(k) = \int_{\mathopen]k,a]} G(x)\,dF(x) +
    \int_{\mathopen]k,a]} F(x)\,dG(x),
  \]
  i.e.
  \[
    (a-k)F(a) = \int_{\mathopen]k,a]} (x-k)\,dF(x) + \int_k^a
    F(x)\,dx,
  \]
  which is equivalent to
  \[
    \int_{[k,a]} (x-k)\,dF(x) = \int_k^a (F(a)-F(x))\,dx.
  \]
  Therefore, by the monotone convergence theorem,
  \begin{align*}
  \lim_{a\to\infty} \int_{[k,a]} (x-k)\,dF(x)
  &= \lim_{a\to\infty} \int_{\erre_+} 1_{[k,a]} (x-k)\,dF(x)\\
  &= \int_{[k,\infty\mathclose[} (x-k)\,dF(x)\\
  &= \int_{\erre_+} (x-k)^+\,dF(x) = C(k),
  \end{align*}
  as well as
  \[
  \lim_{a\to\infty} \int_k^a (F(a)-F(x))\,dx
  = \lim_{a\to\infty} \int_{\erre_+} 1_{[k,a]} (F(a)-F(x))\,dx
  = \int_k^\infty (F(\infty)-F(x))\,dx,
  \]
  hence
  \begin{equation}
    \label{eq:C}
    C(k) = \int_k^\infty (F(\infty)-F(x))\,dx.
  \end{equation}
  This implies $C'=F-F(\infty)$ a.e. as well as, by right continuity
  of $F$ and convexity of $C$, $D^+C(x)=F(x)-F(\infty)$ for every
  $x \in \erre_+$.
\end{proof}

The finiteness of the integral on the right-hand side of \eqref{eq:C}
is implied by the finiteness of $C(k)$, which in turn follows by the
assumption that $dF$ has finite mean. One may also easily see directly
that the last assumption implies that the integral is finite. In fact,
this produces another proof of the identity \eqref{eq:C}: by Tonelli's
theorem,
\begin{align*}
  \int_k^\infty (F(\infty)-F(x))\,dx
  &= \int_k^\infty \int_{\mathopen]x,\infty\mathclose[} dF(y)\,dx
    = \int_{[k,\infty\mathclose[} \int_k^y dx\,dF(y)\\
  &= \int_{[k,\infty\mathclose[} (k-y)\,dF(y)
    = \int_{\erre_+} (k-y)^+\,dF(y)\\
  &= C(k).
\end{align*}
The relation between $C$ and $F$ can of course be obtained also from
put-call parity, once the relation between $P$ and $F$ has been
obtained: if follows from the identity $x-k = (x-k)^+ - (k-x)^+$, upon
integrating with respect to $dF$, that
\[
  \int_{\erre_+} x\,dF(x) - k\int_{\erre_+} dF = C(k) - P(k),
\]
hence, by Lebesgue's differentiation theorem,
$-F(\infty) = C'(k) - P'(k) = C'(k) - F(k)$ for a.a. $k \in \erre_+$,
as well as $D^+C(k)=F(k)-F(\infty)$ for every $k \in \erre_+$ by the
same argument used above.

\subsection{Reconstruction of $F$ by approximation in $L^1(dF)$}
Let $V$ be the vector space generated by put payoff profiles, i.e. by
the family of functions $\erre_+ \ni x \mapsto (k-x)^+$, $k \geq
0$. We are going to show the following approximation result.
\begin{lemma}
  Let $a>0$. For any $\varepsilon>0$ there exists $\phi \in V$ such
  that
  \[
  \norm[\big]{\phi - 1_{[0,a]}}_{L^1(dF)} < \varepsilon.
  \]
\end{lemma}
\begin{proof}
  Since $F$ is right-continuous, there exists $b>a$ such that
  $F(b)-F(a)<\varepsilon$.  Set $\phi_a(x):=(a-x)^+$,
  $\phi_b(x):=(b-x)^+$, $\alpha = 1/(b-a)$, and $\phi := \alpha\phi_b
  - \alpha\phi_a$. Then easy computations show that $\phi \colon
  \erre_+ \to [0,1]$ is a continuous function with support $[0,b]$,
  equal to one on $[0,a]$. More precisely,
  \[
  \phi(x) =
  \begin{cases}
    \alpha(b-a)=1, & 0 \leq x \leq a,\\
    \alpha b - \alpha x, & a \leq x \leq b,\\
    0, & x \geq b.
  \end{cases}
  \]
  Since $\phi = 1_{[0,a]} + \phi 1_{\mathopen]a,b]}$, we have
  \[
  \abs[\big]{\phi - 1_{[0,a]}} = \phi 1_{\mathopen]a,b]}
  \leq 1_{\mathopen]a,b]},
  \]
  hence
  \[
  \norm[\big]{\phi-1_{[0,a]}}_{L^1(dF)}
  \leq \int_{[0,b]} 1_{\mathopen]a,b]} \,dF = F(b)-F(a) < \varepsilon.
  \qedhere
  \]
\end{proof}

This shows that we can explicitly approximate $F$ by $P$. The (proof
of the) lemma also shows that $D^+P=F$: for any $a>0$, take a sequence
$(b_n)$ converging to $a$ from the right, and call $\phi_n$ the
corresponding approximating sequence converging to $1_{[0,a]}$ in
$L^1(dF)$, for which
\[
  \int \phi_n\,dF = \int \frac{1}{b_n-a} \bigl( (b_n-x)^+ - (a-x)^+
  \bigr)\,dF(x) = \frac{P(b_n)-P(a)}{b_n-a},
\]
hence
\[
  F(a) = \lim_{n \to \infty} \int \phi_n\,dF
  = \lim_{n \to \infty} \frac{P(b_n)-P(a)}{b_n-a} = D^+P(a).
\]
This approach to proving that $D^+P=F$ (that, by the way, does not
require any further condition on $F$) is probably the most
elementary. Note that the approach via integration by parts of the
previous subsection also implies
\[
F(a) = D^+P(a) = \lim_{n \to \infty} \frac{P(b_n)-P(a)}{b_n-a}
= \lim_{n \to \infty} \int \phi_n\,dF,
\]
while here we prove the seemingly more precise limiting relation
$\phi_n \to 1_{[0,a]}$ in $L^1(dF)$. This, however, can be deduced
from F.~Riesz's lemma:\footnote{This result is often called
  Scheff\'e's lemma: in a general measure space with measure $\mu$, if
  $f_n \to f$ $\mu$-a.e. and $\int \abs{f_n} \,d\mu \to \int
  \abs{f}\,d\mu$, then $f_n \to f$ in $L^1(\mu)$.} since both
$1_{[0,a]}$ and $\phi_n$ are positive, $\phi_n \to 1_{[0,a]}$ a.e. and
$\int \phi_n \,dF \to \int 1_{[0,a]}\,dF$, it follows that $\phi_n \to
1_{[0,a]}$ in $L^1(dF)$. Therefore also the integration-by-parts proof
of $D^+P=F$, together with F.~Riesz's lemma, implies that indicator
functions of intervals can be obtained as limits in the $L^1(dF)$ norm
of linear combinations of put payoff profiles, that are explicitly
determined.

\medskip

Even though the previous lemma is enough to obtain $F$ from $P$, a
more general denseness result holds.
\begin{prop}
  The vector space $V$ generated by put payoff profiles is dense in
  $L^1(dF)$.
\end{prop}
\begin{proof}
  Let $g \in L^1(dF)$ and $\varepsilon>0$. Then there exists $n \in
  \enne$ and $A_i:=\mathopen]a_i,b_i]$, $0 \leq a_i\leq b_i$, and $c_i
  \in \erre$, $i=1,\ldots,n$, such that
  \[
  \norm[\Big]{g - \sum_{i=1}^n c_i1_{A_i}}_{L^1(dF)} \leq
  \varepsilon/2.
  \]
  By the previous lemma, the indicator function of any interval of
  $\erre_+$ open to the left and closed to the right can be
  approximated by an element of $V$. Therefore, for every
  $i=1,\ldots,n$ there exists $\phi_i \in V$ such that (all norms
  until the end of the proof are meant to be in $L^1(dF)$)
  \[
  \norm{\phi_i - 1_{A_i}} \leq \frac{1}{n\abs{c_i}}
  \frac{\varepsilon}{2},
  \]
  hence, setting $\phi:=\sum c_i\phi_i$,
  \begin{align*}
  \norm{g-\phi}
  &\leq \norm[\Big]{g - \sum_{i=1}^n c_i1_{A_i}}
    + \sum_{i=1}^n \abs{c_i} \norm[\big]{ 1_{A_i} - \phi_i}\\
  &\leq \varepsilon/2 + \sum_{i=1}^n \abs{c_i} \frac{1}{n\abs{c_i}}
  \frac{\varepsilon}{2} = \varepsilon.
  \end{align*}
  Since $\phi$ clearly belongs to $V$, the proof is completed.
\end{proof}

\ifbozza\newpage\else\fi
\section{Convex payoffs}
\label{sec:conv}
We are going to show that prices of call options for all strikes
determine the price of any option with arbitrary convex payoff
function (the result is not new -- see, e.g.,
\cite[pp.~51-52]{KS:MMF}, with a different proof), thus also for
options with payoff function that can be written as the difference of
two convex functions.

Using the language of \S\ref{sec:pf}, let $M_1=(g,dF(g))_{g \in G}$ be
the measurement set with $G$ the set of convex functions on $\erre_+$
(satisfying the assumption below), and
$M_2=(g_k,dF(g_k))_{k\in\erre_+}$, $g_k \colon x \mapsto (x-k)^+$, the
measurement set of call options (for all strikes).  It is evident that
$M_1$ is finer than $M_2$. We shall show that $M_2^m$ is finer than
$M_1$, hence that $M_1$ and $M_2^m$ are equivalent (in particular,
$M_1$ is a representation). The proof will actually establish that,
for any $g \in G$, $\pi(g)$ can be written in terms of an integral of
the function $C \colon k \mapsto \pi(g_k)$. This will then be shown to
belong to $M_2^m$.

\medskip

Throughout this section we assume that $g \colon \erre_+ \to \erre$ is
the restriction to $\erre_+$ of a convex function $h$ on an open set
$I \supset \erre_+$. In particular, $D^+g(0)>-\infty$. In order to
avoid trivialities, we also assume that $g \in L^1(dF)$. We recall
that $g$ is continuous, differentiable almost everywhere,
right-differentiable on $\mathopen[0,\infty\mathclose[$, and that
$D^+g$ is increasing and c\`adl\`ag. In particular, $D^+g$ has finite
variation, thus generates a Lebesgue-Stieltjes measure that we shall
denote by $m$, or also by $dg'$. The positive measure $m$ can also be
identified with the second derivative of $g$ in the sense of
distributions.
\begin{prop}
  \label{prop:cvp}
  Assume that $\overline{dF}<\infty$ and let
  $C \colon \erre_+ \to \erre_+$ be the call option price
  function. Then
  \begin{equation}
    \label{eq:cvp}
    \int_{\erre_+} g\,dF = g(0)F(\infty)
    + D^+g(0) \overline{dF} + \int_{\mathopen]0,\infty\mathclose[} C\,dm.
  \end{equation}
\end{prop}
\begin{proof}
  We have
  \[
  g(x) = g(0) + \int_0^x D^+g(y)\,dy, 
  \]
  where $D^+g(y) - D^+g(0) = m(\mathopen]0,y\mathclose])$ for every $y>0$, 
  hence, by Tonelli's theorem,
  \begin{align*}
    g(x) &= g(0) + D^+g(0)x + \int_0^x \int_{\mathopen]0,y]} dm(k)\,dy\\
         &= g(0) + D^+g(0)x
    + \int_{]0,\infty[} \int_{[k,x]} dy \,dm(k)\\
    &= g(0) + D^+g(0)x + \int_{]0,\infty[} (x-k)^+\,dm(k).
  \end{align*}
  Integrating both sides with respect to $dF$ and appealing again to
  Tonelli's theorem completes the proof.
\end{proof}

\noindent Note that
\[
  \int_{[0,\infty\mathclose[} C\,dm = C(0)m(\{0\}) +
  \int_{\mathopen]0,\infty\mathclose[} C\,dm
  = D^+g(0) \overline{dF} + \int_{\mathopen]0,\infty\mathclose[} C\,dm,
\]
i.e. \eqref{eq:cvp} could be written in the more symmetric form
\[
    \int_{\erre_+} g\,dF = g(0)F(\infty)
    + \int_{\erre_+} C\,dm.
\]
Analogously, since
\[
\int_{[0,\infty\mathclose[} g\,dF = g(0)F(0)
+ \int_{\mathopen]0,\infty\mathclose[} g\,dF,
\]
\eqref{eq:cvp} could also be written as
\[
\int_{\mathopen]0,\infty\mathclose[} g\,dF = g(0) (F(\infty)-F(0))
+ D^+g(0) \overline{dF} + \int_{\mathopen]0,\infty\mathclose[} C\,dm.
\]

\begin{coroll}
  Let $I \subseteq \erre$ be an open set containing $\erre_+$ and
  $h_1$, $h_2 \colon I \to \erre_+$ convex functions belonging to
  $L^1(dF)$. If $g=h_1-h_2$ and $\nu$ is the Lebesgue-Stieltjes
  (signed) measure induced by $D^+h_1-D^+h_2$,
  i.e. $\nu([0,x]):=D^+h_1(x)-D^+h_2(x)$, then
  \[
    \int_{\erre_+} g\,dF = g(0)F(\infty)
    + D^+g(0) \overline{dF} + \int_{\mathopen]0,\infty\mathclose[} C\,d\nu.
  \]
\end{coroll}

\noindent Slightly more generally, one can also write
\[
\beta_T^{-1} g(S_T) = g(0)\beta_T^{-1} 
+ D^+g(0)\beta_T^{-1}S_T + \int_{]0,\infty[} \beta_T^{-1}(S_T-k)^+\,d\nu(k),
\]
hence, taking conditional expectation with respect to $\cF_t$, for any
$t \in [0,T]$, and multiplying by $\beta_t$,
\begin{align*}
\beta_t \E_\Q\bigl[\beta_T^{-1} g(S_T)\big\vert\cF_t\bigr]
&= g(0) \beta_t \E_\Q\bigl[\beta_T^{-1} \big\vert\cF_t\bigr]\\
&\quad + D^+g(0) \beta_t \E_\Q\bigl[\beta_T^{-1} S_T \big\vert\cF_t\bigr]\\
&\quad + \int_{]0,\infty[} \beta_t \E_\Q\bigl[\beta_T^{-1} (S_T-k)^+ \big\vert
\cF_t\bigr]\,d\nu(k)\\
&= g(0) B(t,T) + D^+g(0) S_t + \int_{]0,\infty[} C_t(k)\,d\nu(k),
\end{align*}
where
$C_t(k) := \beta_t \E_\Q\bigl[\beta_T^{-1} (S_T-k)^+ \big\vert
\cF_t\bigr]$ is the price at time $t$ of the call option with strike
$k$.

It is actually possible to prove Proposition~\ref{prop:cvp} using only
the integration by parts formula \eqref{eq:ibp}. Even though the proof
is longer than the previous one, some of its ingredients may be
interesting in their own right.
We begin with a useful reduction step.
\begin{lemma}
  \label{lm:red}
  Assume that $\overline{dF}<\infty$. The claim of
  Proposition~\ref{prop:cvp} holds if and only if it does under the
  additional assumptions that $g(0)=D^+g(0)=0$ and $m$ has compact
  support.
\end{lemma}
\begin{proof}
  Clearly only sufficiency needs a proof. The extra assumption
  $g(0)=D^+g(0)=0$ comes at no loss of generality as one can reduce to
  this situation simply replacing the function $g$ by the function $x
  \mapsto g(x) - g(0) - D^+g(0)x$, which is still convex, being the
  sum of a convex function and an affine function, as well as in
  $L^1(dF)$, because $\overline{dF}$ is finite by assumption. Let us
  then assume that $g(0)=D^+g(0)=0$.
  Let $(\chi_n)$ be a sequence of smooth cutoff functions such that
  $\chi_n \colon \erre_+ \to [0,1]$ has support equal to $[0,n+1]$ and
  is equal to one on $[0,n]$. Setting, for every $n \in \enne$,
  $m_n:=\chi_n m$ and
  \[
    g^{(1)}_n(x) := m_n([0,x]) = 
    \int_{[0,x]} \chi_n\,dm, \qquad
    g_n(x) := \int_0^x g^{(1)}_n(y)\,dy,
  \]
  it is immediately seen that $g^{(1)}_n$ is positive,
  $g_n'=g_n^{(1)}$ a.e. and $D^+g_n=g_n^{(1)}$, and $g_n$ is convex.
  Therefore, by hypothesis,
  \[
    \int_{\erre_+} g_n\,dF = \int_{\mathopen]0,\infty\mathclose[} C \,dm_n
    = \int_{\mathopen]0,\infty\mathclose[} C \chi_n\,dm.
  \]
  Several applications of the monotone convergence theorem imply that
  $(g_n)$ converges pointwise from below to $g$, hence, finally, that
  \[
  \int_{\erre_+} g\,dF = \int_{\mathopen]0,\infty\mathclose[} C \,dm.
  \qedhere
  \]
\end{proof}

\noindent Note that the ``normalizing'' assumptions $g(0)=0$ and
$D^+g(0)=0$ imply that
\[
\int_{\erre_+} g\,dF = \int_{\mathopen]0,\infty\mathclose[} g\,dF
\]
and that
\[
\int_{\mathopen]0,\infty\mathclose[} C\,dm = \int_{\erre_+} C \,dm,
\]
respectively. The former is evident, and the latter follows from
$m(\{0\})=D^+g(0)=0$. Therefore
\[
\int_0^\infty g\,dF = \int_{\mathopen]0,\infty\mathclose[} g\,dF
= \int_{\mathopen]0,\infty\mathclose[} C\,dm = \int_{\erre_+} C \,dm.
\]

\begin{proof}[An alternative proof of Proposition~\ref{prop:cvp}]
  We shall assume, as the previous lemma allows to do, that
  $g(0)=D^+g(0)=0$ and that $m$ has compact support, which implies
  that, for all $x$ sufficiently large, $g$ is differentiable at $x$
  and $g'(x)$ is constant. For the rest of the proof, we shall write,
  with a harmless abuse of notation, $g'$ to denote $D^+g$.
  Since $g$ is continuous and $F$ is c\`adl\`ag, the integration by
  parts formula \eqref{eq:ibp} yields, for any $a \in \erre_+$,
  \[
  g(a) F(a) - g(0)F(0) = \int_{\mathopen]0,a\mathclose]} g(x)\,dF(x) 
  + \int_{\mathopen]0,a\mathclose]} F(x)\,dg(x).
  \]
  Therefore, as $g(0)=0$ and the Lebesgue-Stieltjes measure $dg$ is absolutely
  continuous with respect to Lebesgue measure with density $g'$,
  \[
  \int_{[0,a]} g(x)\,dF(x) = g(a) F(a) 
  - \int_0^a g'(x) F(x)\,dx,
  \]
  hence
  \[
  \int_0^\infty g(x)\,dF(x) = \lim_{a\to+\infty} \Bigl( g(a) F(a) -
  \int_0^a g'(x) F(x)\,dx \Bigr).
  \]
  Since $g'$ is increasing and c\`adl\`ag, and $C$ is continuous, another
  application of the integration by parts formula \eqref{eq:ibp}
  yields, for any $a \in \erre_+$,
  \[
  g'(a)C(a) - g'(0)C(0) = \int_0^a g'(x)\,dC(x) 
  + \int_{\mathopen]0,a\mathclose]} C(x)\,dg'(x),
  \]
  hence, recalling that $g'(0)=0$,
  \[
  \int_0^a g'(x)\,dC(x) = g'(a)C(a) - \int_{\mathopen]0,a]} C\,dm.
  \]
  Moreover, the identity $C'=F-F(\infty)$ a.e. implies
  \[
    \int_0^a g'(x)\,dC(x) = -F(\infty)g(a)
    + \int_0^a g'(x)F(x)\,dx,
  \]
  hence
  \begin{align*}
    -\int_0^a g'(x)F(x)\,dx
    &= -F(\infty)g(a) - \int_0^a g'(x)\,dC(x)\\
    &= -F(\infty)g(a) - g'(a)C(a) + \int_{\mathopen]0,a]} C\,dm,
  \end{align*}
  thus also
  \[
  \int_{\erre_+} g(x)\,dF(x) = \lim_{a\to+\infty} \Bigl(
  g(a)(F(a)-F(\infty)) - g'(a)C(a) + \int_{\mathopen]0,a]} C\,dm \Bigr).
  \]
  Note that $g'$ is increasing by convexity of $g$ and $g'(0)=0$,
  hence $g'$ is positive, therefore $g$ is increasing and positive
  because $g(0)=0$. Therefore
  \[
  \abs[\big]{g(a) (F(a)-F(\infty))} = g(a) (F(\infty)-F(a)) =
  \int_{\mathopen]a,+\infty\mathclose[} g(a)\,dF 
  \leq \int_{\mathopen]a,+\infty\mathclose[} g(x)\,dF(x),
  \]
  where the last term converges to zero as $a \to +\infty$ because $g
  \in L^1(dF)$ by assumption. In particular,
  \[
  \lim_{a \to +\infty} g(a)(F(a)-F(\infty)) = 0.
  \]
  Moreover, as $g'$ is constant at infinity and $C$ tends to zero at
  infinity, we also have
  \[
  \lim_{a \to +\infty} g'(a) C(a) = 0,
  \]
  which allows to conclude that
  \[
  \int_{\erre_+} g\,dF = \int_{\mathopen]0,\infty\mathclose[} C\,dm.
  \qedhere
  \]
\end{proof}

Let us show that
$\int_{\mathopen]0,\infty\mathclose[} C\,dm \in M_2^m$. By Tonelli's
theorem,
\begin{align*}
  \int_{\mathopen]0,\infty\mathclose[} C\,dm
  &= \int_{\mathopen]0,\infty\mathclose[}
    \int_{\erre_+} (x-k)^+ \,dF(x)\,dm(k)\\
  &= \int_{\erre_+}
    \int_{\mathopen]0,\infty\mathclose[} (x-k)^+ \,dm(k)\,dF(x).
\end{align*}
Let $(k_i)_{i=0,\ldots,2^n}$ be a dyadic partition of
$\mathopen]0,n]$. Then
\[
  \sum_{i=1}^{2^n} (x-k_{i+1})^+
  1_{\mathopen]k_1,k_{i+1}\mathclose]}(k) \uparrow (x-k)^+
  \qquad \forall x, \,  k \in \erre_+
\]
as $n \to \infty$, hence, again by Tonelli's theorem,
\begin{align*}
  &\int_{\mathopen]0,\infty\mathclose[} \sum_{i=1}^{2^n} (x-k_{i+1})^+
    1_{\mathopen]k_1,k_{i+1}\mathclose]}(k)\,dm(k)\\
  &\hspace{3em} = \sum_{i=1}^{2^n} m\bigl(\mathopen]k_i,k_{i+1}\mathclose]\bigr)
    (x-k_{i+1})^+
    \, \uparrow \, \int_{\mathopen]0,\infty\mathclose[} (x-k)^+\,dm(k)
    \qquad \forall x \in \erre_+.
\end{align*}
Then
\[
g_n := \sum_{i=1}^{2^n} m\bigl(\mathopen]k_i,k_{i+1}\mathclose]\bigr)
    (x-k_{i+1})^+
\]
defines a sequence of elements in the vector space generated by $M_2$
that monotonically converges pointwise to the function
$x \mapsto \int_{\mathopen]0,\infty\mathclose[} (x-k)^+\,dm(k)$, which
belongs to $L^1(dF)$ by assumption, therefore also to $M_2^m$.

\begin{rmk}
  It is more convenient to work with the call price function $C$,
  rather than with the put price function $P$, because $C$ vanishes at
  infinity, while $P$ grows linearly at infinity (see
  Propositions~\ref{prop:P} and \ref{prop:C}). However, the identity
  \[
  x-k = (x-k)^+ - (x-k)^- = (x-k)^+ - (k-x)^+
  \]
  yields, upon integrating both sides with respect to $dF$,
  \[
    \int_{\erre_+} x\,dF(x) - k \int_{\erre_+} dF(x)
    = \overline{dF} - kF(\infty) = C(k) - P(k),
  \]
  i.e.
  \begin{equation}
    \label{eq:PC}
    C(k) = P(k) - kF(\infty) + \overline{dF},
  \end{equation}
  hence $k \mapsto P(k) - kF(\infty) + \overline{dF} \in L^1(m)$, even
  though, in general, $P$ need not belong to $L^1(dF)$. A formula
  relating the integral of $g$ with respect to $dF$ with the integral
  of $P$ with respect to $m$ for a special class of functions $g$ will
  be discussed in the next section.
\end{rmk}

\begin{rmk}
  A small variation of the argument used in the proof of
  Lemma~\ref{lm:red} shows that every $C^2$ function $g$ is the
  difference of two convex functions $h_1$ and $h_2$ (taking the
  positive and negative part of $g''$). A simple sufficient condition
  ensuring that the functions $h_1$ and $h_2$ can be chosen in
  $L^1(dF)$ is that there exists a function $h \in L^1(dF)$ with
  $h''=\abs{g''}$.
\end{rmk}

\ifbozza\newpage\else\fi
\section{A distributional approach}
\label{sec:distrib}
We are going to show that most properties of the functions $F$, $P$
and $C$ discussed in the previous sections can also be obtained using
Schwartz's distributions. The main advantage of this approach is that
several results reduce, in the formal aspect, to simple calculus for
distributions. Some work is needed, however, to remove the regularity
assumptions on test functions typical of this approach.

Throughout this section, the functions $F$, $P$, and $C$ (the last one
if $\overline{dF}$ is finite) are assumed to be extended to $\erre$
setting them equal to zero on $\mathopen]-\infty,0\mathclose[$. All of
them are obviously locally in $L^1(\erre)$, hence they can be
considered as distributions in $\cD'(\erre)$. For instance,
\[
\ip{F}{\phi} := \int_\erre F(x)\phi(x)\,dx, \qquad \phi \in \cD(\erre)
\]
(as is customary, we shall use the same symbols to define both a
functions and the corresponding distribution).  Moreover, the measure
$dF$ can be identified with the distributional derivative $F'$ of
$F$. In fact, for any $\phi \in \cD(\erre)$,
\[
\ip{F'}{\phi} = - \ip{F}{\phi'} = - \int_\erre F(x)\phi'(x)\,dx,
\]
where, thanks to the integration-by-parts formula,
\[
- \int_\erre F(x)\phi'(x)\,dx 
= \int_{[0,\infty\mathclose[} \phi(x)\,dF(x).
\]
The price function for put options $P$ can be written in terms of
convolutions of distributions. In fact, denoting the function
$x \mapsto x^+$ by $(\cdot)^+$, the function $P$ is the convolution of
$(\cdot)^+$ with the measure $dF$, therefore, since $dF=F'$ in $\cD'$
and both $F'$ and $(\cdot)^+$, interpreted as elements of $\cD'$,
are supported on $\erre_+$, the convolution of $(\cdot)^+$ and $F'$ is
well-defined in the sense of distributions, and
$P = (\cdot)^+ \ast F'$ in $\cD'$. Let $H=1_{\erre_+}$ denote the
(right-continuous) Heaviside function. Standard calculus in $\cD'$
then yields
\[
P = (\cdot)^+ \ast F' = \bigl((\cdot)^+\bigr)' \ast F = H \ast F,
\]
thus also, denoting the Dirac measure at the origin by $\delta$,
\[
P' = H' \ast F = \delta \ast F = F,
\]
and $P'' = F'$, all as identities in $\cD'$. As already observed, $F'$
coincides with the Lebesgue-Stieltjes measure $dF$, hence it is a
positive distribution. As is well known, a distribution with positive
second derivative is a convex function, hence we recover the convexity
of $P$. This and the identity $P'=F$ in $\cD'$ also imply that $P'=F$
holds also in the a.e. sense in $\erre$, and that one can choose a
right-continuous version of $P'$, so that $D^+P=F$.
The properties of $P$ have thus been obtained starting from the
properties of its second distributional derivative, that is reversing
the path followed in the previous section, where convexity of $P$ was
proved first, which implied first-order differentiability outside a
countable set of points first, hence second-order
differentiability in the sense of measures.

On the other hand, it seems not possible to treat the call option
price function $C$ by similar arguments, because one would formally
have $C=(\cdot)^- \ast F'$, where, however, the convolution is not
well-defined in the sense of distributions. In fact, $(\cdot)^-$ and
$F'$ do not have their support ``on the same side'' of $\erre$, and
none of them has compact support. Nonetheless, properties of $C$ can
be deduced from those of $P$ taking \eqref{eq:PC} into account. Since
we are considering $C$ and $P$ as distributions on $\erre$, it is
convenient to rewrite \eqref{eq:PC} as
\[
C = P - F(\infty)(\cdot)^+ + \overline{dF}H,
\]
which can be interpreted both as an identity of c\`adl\`ag functions
on $\erre$, as well as an identity in $\cD'(\erre)$. Differentiating
in $\cD'(\erre)$ yields
\begin{align}
  C'  &= P' - F(\infty)H + \overline{dF} \delta,\nonumber\\
  C'' &= P'' - F(\infty)\delta + \overline{dF} \delta'.\label{eq:C2}
\end{align}
Since $\delta'$ is not a measure, the function $C$ is \emph{not}
convex on $\erre$ (this also clearly follows from $C(0)=\overline{dF}$
and $C(k)=0$ for all $k<0$). On the other hand, one also infers that
\[
  C' = P' - F(\infty), \quad C'' = P'' \qquad
  \text{ in } \cD'(\mathopen]0,\infty\mathclose[),
\]
hence $C$ is convex on $\mathopen]0,\infty\mathclose[$, then also on
$\erre_+$, and one can choose a right-continuous version of $C'$ on
$\erre_+$ such that $C'(x)=F(x)-F(\infty)$ for a.a. $x \in \erre_+$,
with $D^+C(x) = F(x)-F(\infty)$ for every $x \in \erre_+$.

\medskip

We are now going to show how to prove \eqref{eq:cvp} using
distribution arguments. Note that, assuming that $g(0)=Dg^+(0)=0$,
\eqref{eq:cvp} could heuristically be written as
$\ip{C''}{g}=\ip{C}{g''}$, which seems very natural indeed. It is
clear, however, that it makes no sense if $g$ is just a convex
function in $L^1(dF)$. However, note that the identity has a meaning
if $\ip{\cdot}{\cdot}$ is interpreted as the duality between measures
and continuous functions, rather than between distributions and test
functions.

Let us start from the identity $P''=F'=dF$ in $\cD'(\erre)$ that was
proved above. Then we immediately obtain $\ip{P''}{g} = \ip{P}{g''}$
for every $g \in \cD(\erre)$, hence also, since $P''$ is a
distribution of order at most two,
\[
\ip{P''}{g} = \int_{\erre_+} g\,dF = \ip{P}{g''} 
\qquad \forall g \in C^2_c(\erre).
\]
Therefore, using identity \eqref{eq:C2},
\begin{align*}
  \ip{P''}{g} = \int_{\erre_+} g\,dF
  &= \ip{C''}{g} + F(\infty)\ip{\delta}{g}
    - \overline{dF}\ip{\delta'}{g}\\
  &= \ip{C''}{g} + F(\infty) g(0) + \overline{dF} g'(0).
\end{align*}
We have thus obtained \eqref{eq:cvp} under the assumption $g \in
C^2_c(\erre)$, or, equivalently, $g \in C^2(\erre_+)$ such that $g(x)=0$
for $x$ sufficiently large.

Let us now assume that $g \in C^2(\erre)$.  As discussed above, we can
and shall assume, without loss of generality, that $g(0)=g'(0)=0$.
Let $(\chi_n)$ be a sequence of smooth cutoff function taking values
in $[0,1]$, equal to one on $[-a,a]$, and equal to zero on
$[a+1/n,\infty\mathopen[$. Then $g\chi_n \in C^2_c(\erre)$ and
\[
(g\chi_n)'' = g''\chi_n + 2g'\chi_n' + g\chi_n'',
\]
hence
\begin{align*}
\int_{\erre_+} g\chi_n \,dF = \ip{C''}{g\chi_n} 
&= \ip{C}{(g\chi_n)''}\\
&= \ip{C}{g''\chi_n} + 2\ip{C}{g'\chi_n'} + \ip{C}{g\chi_n''}.
\end{align*}
We are going to pass to the limit as $n \to \infty$.
One has
\[
\int_{\erre_+} g\chi_n \,dF = \int_{[0,a]} g\,dF
+ \int g\chi_n 1_{\mathopen]a,a+1/n\mathopen[} \,dF,
\]
where $g\chi_n 1_{\mathopen]a,a+1/n\mathopen[} \to 0$ pointwise,
hence, by the dominated convergence theorem,
\[
\lim_{n \to \infty} \int_{\erre_+} g\chi_n \,dF = \int_{[0,a]} g \,dF.
\]
An entirely similar, slightly simpler reasoning shows that
\[
\lim_{n \to \infty} \ip{C}{g''\chi_n} 
= \lim_{n \to \infty} \int_{\erre_+} C g'' \chi_n
= \int_0^a C g''.
\]
Moreover,
\[
\ip{C}{g'\chi_n'} = \int_{\erre} C(x)g'(x)\chi'_n(x)\,dx
= \int_a^{a+1/n} C(x)g'(x)\chi'_n(x)\,dx,
\]
where $-\chi'_n$ converges to $\delta_a + R$ in the sense of
distributions, where $\delta_a$ is the Dirac measure at $a$ and $R$ a
distribution with support contained in
$\mathopen]-\infty,-a]$, hence
\[
\lim_{n \to \infty} \ip{C}{g'\chi_n'} = -C(a)g'(a).
\]
The term $\ip{C}{g\chi_n''}$ is more difficult to treat because
$\chi_n''$ converges to $-\delta_a'$ in the sense of distributions
(modulo terms with support in the strictly negative reals, that we are
going to ignore), but $C$ is just right-differentiable, not of class
$C^1$. We can nonetheless argue as follows: let $(\rho_m)$ be a
sequence of mollifiers with support contained in $[-1/m,0]$ and set
$C_m := C \ast \rho_m$. Then $C_m \in C^\infty(\erre)$ and
\[
\ip{C_m}{g\chi_n''} = \ip{C_mg}{\chi_n''} = - \ip{(C_mg)'}{\chi_n'},
\]
hence
\[
\lim_{n \to \infty} \ip{C_m}{g\chi_n''} = C_m'(a)g(a) + C_m(a)g'(a).
\]
Thus one has
\[
\int_0^a C_m''g = \int_0^a C_mg'' + C'_m(a)g(a) - C_m(a)g'(a).
\]
We can now pass to the limit as $m \to \infty$: the continuity of $C$
implies that $C_m$ converges to $C$ uniformly on $[0,a]$, hence
\[
\lim_{m \to \infty} \int_0^a C_mg'' = \int_0^a Cg'', \qquad
\lim_{m \to \infty} C_m(a)=C(a).
\]
Setting $dF_m := C''_m = dF \ast \rho_m$ and
$\widetilde{\rho}_m \colon x \mapsto \rho_m(-x)$, so that the support
of $\widetilde{\rho}_m$ is contained in $[0,1/m]$, one has
\[
\int_0^a gC''_m = \int g1_{[0,a]}\,dF_m
= \int g1_{[0,a]} \ast \widetilde{\rho}_m\,dF,
\]
where
\begin{align*}
  \lim_{n \to \infty} g1_{[0,a]} \ast \widetilde{\rho}_m(x)
  &= g(x) \qquad \forall x \in \mathopen]0,a\mathclose[,\\
  \lim_{n \to \infty} g1_{[0,a]} \ast \widetilde{\rho}_m(0)
  &= 0,\\
  \lim_{n \to \infty} g1_{[0,a]} \ast \widetilde{\rho}_m(a)
  &= g(a-) = g(a),
\end{align*}
i.e.
\[
  \lim_{n \to \infty} g1_{[0,a]} \ast \widetilde{\rho}_n(x)
  = g1_{\mathopen]0,a]}(x) \qquad \forall x \in \erre,
\]
or, in other words, $g1_{[0,a]} \ast \widetilde{\rho}_n$ converges to
the c\`agl\`ad version of $g1_{[0,a]}$. Therefore, by the dominated
convergence theorem,
\begin{align*}
  \lim_{m \to \infty} \int_0^a gC''_m
  = \lim_{m \to \infty} \int g1_{[0,a]} \ast \widetilde{\rho}_m\,dF
  = \int_{\mathopen]0,a\mathclose]} g\,dF
  = \int_{[0,a]} g\,dF,
\end{align*}
where the last equality follows from $g(0)=0$.

Since $C_m \in C^\infty(\erre)$ and $C$ is right-differentiable with
increasing incremental quotients (because it is convex), the dominated
convergence theorem yields
\begin{align*}
  C'_m(a) = D^+C_m(a)
  &= \lim_{h \to 0+} \frac{C_m(a+h)-C_m(a)}{h}\\
  &= \lim_{h \to 0+} \int_\erre \frac{C(a-y+h)-C(a-y)}{h} \rho_m(y)\,dy\\
  &= \int_\erre D^+C(a-y) \rho_m(y)\,dy,
\end{align*}
hence also, recalling that the support of $\rho_m$ is contained in
$\erre_-$ and that $D^+C$ is right-continuous,
\[
  \lim_{m \to \infty} C'_m(a) - D^+C(a) =
  \lim_{m \to \infty}  \int_\erre \bigl(D^+C(a-y) - D^+C(a)\bigr) 
  \rho_m(y)\,dy  = 0.
\]
We have thus shown that
\[
\int_{[0,a]} g\,dF = \int_0^a Cg'' - C(a)g'(a) + D^+C(a)g(a)
\]
for every $g \in C^2(\erre)$. To remove the assumption that
$g \in C^2$, assuming instead that it is convex, we can apply the same
regularization by convolution: let $g$ be convex and set
$g_n := g \ast \rho_n$, with the sequence of mollifiers $(\rho_n)$
chosen as before. Then $g_n \in C^\infty$ and
\[
\int_{[0,a]} g_n \,dF = \int_0^a C g_n'' - C(a)g_n'(a) + D^+C(a)g_n(a),
\]
where $g_n \to g$ uniformly on $[0,a]$ and
$\lim_{n \to \infty} g'_n(a) = D^+g(a)$. Moreover, using the same
argument as before,
\[
  \lim_{n \to \infty} \int_0^a C g_n''
  = \lim_{n \to \infty} \int C1_{[0,a]} \ast \widetilde{\rho}_n\,dm
  = \int_{\mathopen]0,a]} C\,dm.
\]
We conclude that
\begin{equation}
  \label{eq:sudo}
  \int_{[0,a]} g\,dF = \int_{\mathopen]0,a]} g\,dF
  = \int_{\mathopen]0,a]} C\,dm - C(a)D^+g(a) + D^+C(a)g(a).
\end{equation}
Note that until here we have not used the assumption that $g \in
L^1(dF)$. To complete the proof of \eqref{eq:cvp}, we let $a$ tend to
infinity using two lemmas proved next, according to which the last two
terms on the right-hand side of \eqref{eq:sudo} tend to zero. It is
precisely at this point that we use the assumption that $g \in
L^1(dF)$.

\begin{lemma}
  Assume that $dF$ is a finite measure and let $g \in L^1(dF)$ be
  increasing. Then
  \[
    \lim_{a \to \infty} g(a)(F(\infty)-F(a)) = 0.
  \]
  In particular, if $\overline{dF}<\infty$ then
  $\lim_{a \to \infty} D^+C(a)g(a)=0$.
\end{lemma}
\begin{proof}
  Assume first that $g(0)=0$, so that $g$ is positive. Then, as $g$ is
  increasing,
  \[
    g(a)(F(\infty)-F(a)) = \int_{\mathopen]a,\infty\mathclose[} g(a) dF(x)
    \leq \int_{\mathopen]a,\infty\mathclose[} g(x) dF(x),
  \]
  and
  \[
  \lim_{a \to \infty} \int_{\mathopen]a,\infty\mathclose[} g(x) dF(x) = 0
  \]
  because $g \in L^1(dF)$.  If $g(0)<0$, then consider the function
  $\tilde{g}:=\abs{g(0)}+g$, which is increasing and belongs to
  $L^1(dF)$. The identity
  \[
    g(a)(F(\infty)-F(a))
    = \tilde{g}(a)(F(\infty)-F(a)) - \abs{g(0)}(F(\infty)-F(a)) = 0
  \]
  immediately implies the claim.
\end{proof}

\begin{lemma}
  Assume that $\overline{dF}<\infty$. Let $g \in L^1(dF)$ be
  absolutely continuous and such that $g'$ is increasing (possibly
  after a suitable modification on a set of Lebesgue measure
  zero). Then
  \[
    \lim_{a \to \infty} g'(a) \int_a^\infty (F(\infty)-F(y))\,dy = 0,
  \]
  or, equivalently, $\lim_{a \to \infty} C(a)g'(a)=0$.
\end{lemma}
\begin{proof}
  The assumption $\overline{dF}<\infty$ guarantees that the function
  $C$ is well-defined and
  \[
  C(a) = \int_a^\infty (F(\infty)-F(y))\,dy \qquad \forall a \in \erre_+.
  \]
  Then we can write
  \begin{align*}
    g'(a)C(a)
    &= \int_a^\infty g'(a) (F(\infty)-F(y)) \,dy\\
    &\leq \int_a^\infty g'(y) (F(\infty)-F(y)) \,dy\\
    &= \int_a^\infty g'(y) \int_{\mathopen]y,\infty\mathclose[} dF(x)\,dy\\
    &= \int_{\mathopen]a,\infty\mathclose[} \int_a^x g'(y)\,dy \,dF(x)
      = \int_{\mathopen]a,\infty\mathclose[} g(x)\,dF(x)
      - \int_{\mathopen]a,\infty\mathclose[} g(a)\,dF(x),
  \end{align*}
  where
  \[
  \lim_{a \to \infty} \int_{\mathopen]a,\infty\mathclose[} g(x)\,dF(x) = 0
  \]
  because $g \in L^1(dF)$. Moreover,
  \[
    \abs[\bigg]{\int_{\mathopen]a,\infty\mathclose[} g(a)\,dF(x)}
    \leq \int_{\mathopen]a,\infty\mathclose[} \abs{g(a)}\,dF(x).
  \]
  Let us first consider the case that $g'(0) \geq 0$, so that $g'$ is
  positive and $g$ is increasing. If there exists $a_0 \in \erre_+$
  such that $g(a_0) \geq 0$, then 
  \[
    \lim_{a \to \infty} \int_{\mathopen]a,\infty\mathclose[} g(a)\,dF(x)
    \leq \lim_{a \to \infty} \int_{\mathopen]a,\infty\mathclose[} g(x)\,dF(x)
    = 0.
  \]
  Otherwise, if $g(x) \leq 0$ for all $x \in \erre_+$, then $\abs{g} = -g$
  is decreasing, therefore
  \[
    \lim_{a \to \infty} \int_{\mathopen]a,\infty\mathclose[} \abs{g(a)}\,dF(x)
    \leq \lim_{a \to \infty} \abs{g(1)} (F(\infty)-F(a)) = 0.
  \]
  Let us now consider the case that $g'(0)<0$: introduce the function
  $\tilde{g}(x):=g(x) + \abs{g'(0)}x$, for which
  $\tilde{g}'(0)=g'(0)+\abs{g'(0)} \geq 0$, and note that $\tilde{g}'$
  is increasing. The assumption $\overline{dF}<\infty$ implies that
  $\tilde{g} \in L^1(dF)$, hence the previous part of the proof shows
  that $\lim_{a \to \infty} C(a)\tilde{g}'(a)=0$. Writing
  $C(a)g'(a) = C(a)\tilde{g}'(a) - C(a) \abs{g'(0)}$ and recalling
  that $\lim_{a \to \infty} C(a)=0$, the proof is completed.
\end{proof}

Even though the set of functions that can be written as the difference
of two convex functions is quite rich (see, e.g., \cite{BaBo}), it
does not contain any discontinuous function. So, for instance, for
digital options we cannot produce a pricing formula such as
\eqref{eq:cvp}. However, the distributional approach allows to obtain
in a quite efficient way pricing formulas for options the payoff of
which can be written piecewise as the difference of convex
functions. The formulas involve, apart from integrals of $C$, also
pointwise evaluations of $C$ and $F$.
Let $g_0 \colon \erre \to \erre$ be a convex function,
$[a,b] \subset \erre_+$ a compact interval, and
$g:=g_01_{\mathopen[a,b\mathclose[}$ a c\`adl\`ag restriction of $g_0$
that will serve as payoff function of an option, of which we are going
to compute the price.

One has
\[
\int_{\erre_+} g\,dF = \int_{\mathopen[a,b\mathclose[} g\,dF 
= \int_{\mathopen]a,b\mathclose[} g\,dF
+ g(a) \Delta F(a),
\]
where $\Delta F(a) = D^+C(a)-D^+C(a-)$ and, by dominated convergence,
\[
  \int_{\mathopen]a,b\mathclose[} g\,dF
  = \lim_{x \to b-} \int_{\mathopen]a,x]} g\,dF.
\]
Since
\[
\int_{\mathopen]a,x]} g\,dF = \int_{[0,x]} g\,dF
- \int_{[0,a]} g\,dF,
\]
it follows by \eqref{eq:sudo} that
\begin{align*}
\int_{\mathopen]a,x]} g\,dF 
&= \int_{\mathopen]a,x]} C\,dm + C(a) D^+g(a) - D^+C(a) g(a)\\
&\quad - C(x) D^+g(x) + D^+C(x) g(x).
\end{align*}
Taking the limit for $x$ going to $b$ from the left, one obtains,
recalling that $C$ is continuous and both $g$ and $D^+C$ are c\`adl\`ag,
\begin{align*}
\int_{\mathopen]a,b\mathclose[} g\,dF
&= \int_{\mathopen]a,b\mathclose[} C\,dm + C(a) D^+g(a) - D^+C(a) g(a)\\
&\quad - C(b) D^+g(b-) + D^+C(b-) g(b-).
\end{align*}

There is an alternative way to obtain the same formula, using
Proposition~\ref{prop:dcd}, that is slightly longer but that starts
from very basic principles and shows how the distributional approach
allows to compute very quickly the price of an option under the very
mild assumption that $F$ admits a continuous density. Let us consider
$g:=g_01_{\mathopen[a,b\mathclose[}$ as a distribution, and assume
first that $F$ is of class $C^1$, which implies that $C$ is of class
$C^2$. Then
\[
\int_{\erre_+} g\,dF = \int_a^b g\,dF = \ip{g}{C''} = \ip{g''}{C},
\]
where, thanks to Proposition~\ref{prop:dcd},
\begin{align*}
\ip{g''}{C} 
&= \int_{\mathopen]a,b\mathclose[} C\,dm + C(a) D^+g(a) - C'(a) g(a)\\
&\quad - C(b) D^+g(b-) + C'(b) g(b-). 
\end{align*}
If $C$ is not twice continuously differentiable, setting $C_n:=C \ast
\rho_n$, with $(\rho_n)$ a sequence of mollifiers chosen as before, then
$C_n$ and $dF_n := dF \ast \rho_n$ are both in $C^\infty$ and
\begin{align*}
  \int_a^b g\,dF_n = \ip{g''}{C_n}
  &= \int_{\mathopen]a,b\mathclose[} C_n\,dm + C_n(a) D^+g(a) - C_n'(a) g(a)\\
  &\quad - C_n(b) D^+g(b-) + C_n'(b) g(b-).
\end{align*}
We are now going to pass to the limit as $n \to \infty$: $C_n$
converges to $C$ uniformly on compact sets, hence $C_n(a)$ and
$C_n(b)$ converge to $C(a)$ and $C(b)$, respectively, and
\[
\lim_{n \to \infty} \int_{\mathopen]a,b\mathclose[} C_n\,dm
= \int_{\mathopen]a,b\mathclose[} C\,dm.
\]
As before, the choice of $(\rho_n)$ and the right continuity of $D^+C$
imply that $C'_n(a)$ and $C'_n(b)$ converge to $D^+C(a)$ and
$D^+C(b)$, respectively, hence
\begin{align*}
  \lim_{n \to \infty} \int_a^b g\,dF_n
  &= \int_{\mathopen]a,b\mathclose[} C\,dm + C(a) D^+g(a) - D^+C(a) g(a)\\
  &\quad - C(b) D^+g(b-) + D^+C(b) g(b-).
\end{align*}
Writing
\[
  \int_a^b g\,dF_n
  = \int g_01_{\mathopen[a,b\mathclose[} \ast \widetilde{\rho}_n\,dF
\]
we can use again an argument already met before, which shows that
\[
  \lim_{n \to \infty} g_01_{\mathopen]a,b\mathclose[} \ast \widetilde{\rho}_n(x)
  = g_-(x) \qquad \forall x \in \erre,
\]
where $g_-$ denotes the c\`agl\`ad version of $g$. Therefore, by
dominated convergence,
\begin{align*}
\lim_{n \to \infty} \int_a^b g\,dF_n 
= \int_{\mathopen]a,b\mathclose]} g_-\,dF
&= \int_{\mathopen]a,b\mathclose[} g_-\,dF + g(b-)\Delta F(b)\\
&= \int_{\mathopen]a,b\mathclose[} g\,dF + g(b-)(D^+C(b)-D^+C(b-)).
\end{align*}
Rearranging terms we are left with
\[
\int_{\mathopen]a,b\mathclose[} g\,dF 
= \int_{\mathopen]a,b\mathclose[} C\,dm + C(a) D^+g(a) - D^+C(a) g(a)
- C(b) D^+g(b-) + D^+C(b-) g(b-),
\]
as before.

\ifbozza\newpage\else\fi
\section{A representation through approximated laws of logarithmic returns}
\label{sec:X}
In the standard Black-Scholes (BS) model one assumes that $S_T =
\exp(\varsigma\sqrt{T} Z -\varsigma^2T/2)$ in law, where the
volatility $\varsigma$ is constant and $Z$ is a standard Gaussian
random variable. This family of random variables (indexed by
$\varsigma$, with time to maturity $T$ fixed as before) can be
embedded in the larger class defined by $S_T = \exp(\sigma X + m)$,
where $\sigma$ and $m$ are constants, and $X$ is a random variable
with density $f \in L^2:=L^2(\erre)$. This rather general family of
laws can be used as setup for empirical non-parametric option pricing,
essentially by projecting the density $f$ on radial basis functions
(see~\cite{cm:Herm}). More precisely, we consider expansions of $f$ in
terms of Hermite functions, so that the lognormal distribution of
returns corresponds exactly to the zeroth order expansion of $f$. The
approach can thus be thought of as a perturbation of the BS model at
fixed time. The following problem then arises: let $(f_n) \subset L^1
\cap L^2$ be a sequence of functions converging to $f$ in $L^2$, and
let $P_n(k)$ be the ``fictitious'' price of a put option with strike
$k$, obtained replacing the density $f$ with its approximation
$f_n$. Suppose that the $P_n(k)$ are known for all $k \in \erre_+$ and
$n \in \enne$. Is this information enough to determine the function
$P$, i.e. the put option prices in the ``true'' model?

Assuming, for simplicity, that $S_0=1$, $\sigma=1$ and $m=0$, and
denoting the distribution function of $X=\log S_T$ with respect to the
measure $\mu$ by $F_X$, it is immediately seen that $F_X(x) = F(e^x)$
for every $x \in \erre$ and that
\[
P(k) = \int_\erre {(k - e^x)}^+\,dF_X(x)
= \int_\erre {(k - e^x)}^+ f(x)\,dx,
\]
hence,
\[
P_n(k) = \int_\erre {(k - e^x)}^+ f_n(x)\,dx.
\]
Note that $F_X$ and $f$ are supported on the whole real line and that
$F$ and $F_X$ are in bijective correspondence, hence $F_X$ is in
bijective correspondence also with $P$. 

The sequence $P_n(k)$ does not necessarily converge to $P(k)$ as $n
\to \infty$, because the function $x \mapsto (k-e^x)^+$ belongs to
$L^\infty$ but not to $L^2$, hence it induces a continuous linear form
on $L^1$, but not on $L^2$. Moreover, convergence in $L^2(\erre)$ does
not imply convergence in $L^1(\erre)$.

We are going to show that the function $P$ can be reconstructed from
approximation to option prices with payoff of the type
\[
  \theta_{k_1,k_2}(x) := (k_2-e^x)^+ - \frac{k_2}{k_1}(k_1-e^x)^+,
  \qquad k_1, k_2 > 0.
\]
More precisely, to identify the pricing functional $P$, it suffices to
know, for any sequence $(f_n)$ converging to $f$ in $L^2$, the values
$\ip{\theta_{k_1,k_2}}{f_n}$ for all $k_1,k_2>0$ and all $n \geq 0$,
where we recall that $\ip{\cdot}{\cdot}$ stands for the scalar product
of $L^2$.

In fact, for any $k_1,k_2>0$, the function $\theta_{k_1,k_2}$ is in
$L^2$, hence, for any sequence $(f_n) \subset L^1 \cap L^2$ converging
to $f$ in $L^2$ (weak convergence in $L^2$ would also suffice), one has
\[
  P_n(k_2) - \frac{k_2}{k_1} P_n(k_1) =
  \ip[\big]{\theta_{k_1,k_2}}{f_n} \longrightarrow
  \ip[\big]{\theta_{k_1,k_2}}{f} = P(k_2) - \frac{k_2}{k_1} P(k_1).
\]
Moreover,
\[
\frac{k_2}{k_1} P(k_1) 
= \int_\erre \frac{k_2}{k_1} (k_1 - e^x)^+ f(x)\,dx,
\]
where $(k_1 - e^x)^+ \in \mathopen]0,k_1]$ for all $x \in \erre$, hence
$\frac{k_2}{k_1} (k_1 - e^x)^+ \in \mathopen]0,k_2]$ for all $x \in \erre$,
and
\[
\frac{k_2}{k_1} (k_1 - e^x)^+ =
\begin{cases}
  \displaystyle k_2 - \frac{k_2}{k_1} e^x, & \text{ if } x \leq \log k_1,\\[8pt]
  0, & \text{ if } x \geq \log k_1,
\end{cases}
\]
hence
\[
\lim_{k_1 \to 0} \frac{k_2}{k_1} (k_1 - e^x)^+ = 0 \qquad \forall x \in \erre.
\]
Therefore the function $x \mapsto \frac{k_2}{k_1} (k_1 - e^x)^+$
converges to zero as $k_1 \to 0$ in $L^p$ for every
$p \in [1,\infty\mathclose[$ by the dominated convergence theorem. In
particular, since $f \in L^2$,
\begin{equation}
  \label{eq:coa}
  \lim_{k_1 \to 0} \frac{k_2}{k_1} P(k_1) =
  \lim_{k_1 \to 0} \int_\erre \frac{k_2}{k_1} (k_1-e^x)^+ f(x)\,dx = 0.
\end{equation}
We have thus shown that
\[
  \lim_{k_1 \to 0} \lim_{n \to \infty} \ip[\big]{\theta_{k_1,k_2}}{f_n}
  = P(k_2) \qquad \forall k_2 > 0,
\]
thus also the following statement.
\begin{prop}
\label{prop:theta}
  Let $(f_n) \subset L^1 \cap L^2$ be a sequence converging to $f$ in
  $L^2$. There is a bijection between
  \[
    \Bigl(\ip[\big]{\theta_{k_1,k_2}}{f_n}\Bigr)_{%
      \begin{subarray}{l}
        k_1,k_2>0 \\ n \geq 0
      \end{subarray}}
  \]
  and $P$.
\end{prop}
\noindent Completely analogously, if $P(k_1)$ is known, then
\[
  P(k_2) = \frac{k_2}{k_1} P(k_1)
  + \lim_{n \to \infty} \ip[\big]{\theta_{k_1,k_2}}{f_n}
  = \frac{k_2}{k_1} P(k_1)
  + \lim_{n \to \infty} \Bigl( P_n(k_2)
  - \frac{k_2}{k_1}P_n(k_1) \Bigr).
\]
\begin{rmk}
  The function $x \mapsto \frac{k_2}{k_1} (k_1 - e^x)^+$ does not
  converge to zero in $L^\infty$ as $k_1 \to 0$, as
  \[
  \sup_{x \in \erre} \frac{k_2}{k_1} (k_1 - e^x)^+ = k_2.
  \]
  However, the convergence in \eqref{eq:coa} also holds with
  $f \in L^1$, i.e. without any extra integrability assumption on $f$,
  because $\frac{k_2}{k_1} (k_1 - e^x)^+ f(x) \leq k_2 f(x)$ for every
  $x \in \erre$, hence the result follows by dominated convergence.
\end{rmk}

Note that Proposition~\ref{prop:theta} can be interpreted as a
representation of $dF$, but, as discussed at the end of
\S\ref{sec:pf}, it cannot be formulated in the language introduced
there. Even the extended measurements of the type $(g_j,\pi_j,dF_j)$,
with $dF_j$ a family of measures weakly converging to $dF$, is not
enough. In fact, it is not difficult to check that the push-forward of
$f_n\,dx$ through $x \mapsto e^x$, denoted by $dF_n$, does not
converge weakly to $dF$, in general. However, setting $M_1 =
(g_k,dF_n(g_k),dF_n)_{k>0,\,n \in \enne}$, where $g_k \colon x \mapsto
(k-x)^+$, we have shown that $M_1$
``implies'' $M_2=(\theta_{k_1,k_2},\pi_{k_1,k_2})_{k_1,k_2>0}$, where
implication is meant as in the last paragraph of \S\ref{sec:pf}, and that
$M_2^m$ is finer than $M$, the measurement set composed of put prices,
which is a representation.

\ifbozza\newpage\else\fi
\bibliographystyle{amsplain}
\bibliography{ref,finanza}

\end{document}